%% file: main.tex
\documentclass[a4paper,UKenglish]{lipics-v2016}
\usepackage{microtype}%if unwanted, comment out or use option "draft"

\usepackage{amsmath}
\usepackage{bm}
\usepackage{stmaryrd}
\usepackage{mathtools}

\usepackage{tikz}
\usetikzlibrary{petri,positioning}

\usepackage{thm-restate}
\usepackage{todonotes}

\usepackage{macros}

\title{Large Flocks of Small Birds: On the Minimal Size of Population Protocols\footnote{M. Blondin was supported by the Fonds de recherche du Quebec – Nature et technologies (FRQNT).}}
\author[1]{Michael Blondin}
\author[2]{Javier Esparza}
\author[3]{Stefan Jaax}
\affil[1]{Technische Universität München, Munich, Germany\\
  \texttt{blondin@in.tum.de}}
\affil[2]{Technische Universität München, Munich, Germany\\
  \texttt{esparza@in.tum.de}}
\affil[3]{Technische Universität München, Munich, Germany\\
  \texttt{jaax@in.tum.de}}
\authorrunning{M. Blondin and J. Esparza and S. Jaax} %mandatory. First: Use abbreviated first/middle names. Second (only in severe cases): Use first author plus 'et. al.'

\Copyright{Michael Blondin, Javier Esparza and Stefan Jaax}%mandatory, please use full first names. LIPIcs license is "CC-BY";  http://creativecommons.org/licenses/by/3.0/

\subjclass{F.1.1 Models of Computation}% mandatory: Please choose ACM 1998 classifications from http://www.acm.org/about/class/ccs98-html . E.g., cite as "F.1.1 Models of Computation". 
\keywords{Population protocols, Presburger arithmetic}% mandatory: Please provide 1-5 keywords
\EventEditors{John Q. Open and Joan R. Acces}
\EventNoEds{2}
\EventLongTitle{42nd Conference on Very Important Topics (CVIT 2016)}
\EventShortTitle{CVIT 2016}
\EventAcronym{CVIT}
\EventYear{2016}
\EventDate{December 24--27, 2016}
\EventLocation{Little Whinging, United Kingdom}
\EventLogo{}
\SeriesVolume{42}
\ArticleNo{23}
\begin{document}

\maketitle

\begin{abstract}
  Population protocols are a well established model of distributed
  computation by mobile finite-state agents with very limited
  storage. A classical result establishes that population protocols
  compute exactly predicates definable in Presburger arithmetic. We
  initiate the study of the minimal amount of memory required to
  compute a given predicate as a function of its size. We present
  results on the predicates $x \geq n$ for $n\in \N$, and more
  generally on the predicates corresponding to systems of linear
  inequalities. We show that they can be computed by protocols with
  $O(\log n)$ states (or, more generally, logarithmic in the
  coefficients of the predicate), and that, surprisingly, some
  families of predicates can be computed by protocols with $O(\log\log
  n)$ states. We give essentially matching lower bounds for the class
  of 1-aware protocols.
\end{abstract}

\section{Introduction}\label{sec:introduction}
\input{sec-introduction}

\section{Preliminaries}\label{sec:preliminaries}
\input{sec-preliminaries}

%\section{Population protocols}\label{sec:protocols}
\input{sec-protocols}

%\subsection{From $k$-way to $2$-way protocols}\label{sec:k:way}
\input{sec-k-way}

\section{Leaderless protocols for $x \geq n$}\label{sec:leaderless}
\input{sec-leaderless}

\section{A $O(\log \log n)$ protocol with leaders for some $x \geq n$}\label{sec:loglog:upper}
\input{sec-loglog-upper}

\section{Universal lower bounds for $1$-aware protocols}\label{sec:one:aware}
\input{sec-one-aware}

\section{Protocols for systems of linear inequalities}\label{sec:sys:lin}
\input{sec-sys-lin}

\section{Conclusion and further work}\label{sec:conclusion}
\input{sec-conclusion}

%% %% Acknowledgements
%% \subparagraph*{Acknowledgements.}

%% We wish to thank...

%% Bibliography
\bibliography{references}

%% Appendix
\clearpage

\appendix

Throughout this appendix, we use the following notation for integer intervals:
For $n, m \in \N$, $n \leq m$, we write $[n, m]$ to denote the set
$\set{n, n + 1, \ldots, m - 1, m}$. Furthermore, by $[n]$ we denote the set $[1, n]$.

\section{Proof of Lemma~\ref{lem:kway}}\label{app:k:way}
\input{app-k-way}

\section{Detailed proofs of Section~\ref{sec:leaderless}}\label{app:leaderless}
\input{app-leaderless}

\section{Detailed proofs of Section~\ref{sec:loglog:upper}}\label{app:loglog:upper}
\input{app-loglog-upper}

\section{Monotonic predicates and 1-awareness}\label{app:monotonic}
\input{app-one-aware-monotonic}

\section{Detailed proofs of Section~\ref{sec:one:aware}}\label{app:one:aware}
\input{app-one-aware}

\section{Detailed proofs of Section~\ref{sec:sys:lin}}\label{app:sys:lin}
\input{app-sys-lin}

%% Temporary
%% \clearpage
%% \section{Javier's stuff}
%\input{javier}

%% \clearpage
%% \section{Michael's stuff}
%\input{michael}

\end{document}

%% file: sec-introduction.tex
Population protocols~\cite{AADFP04} are a model of distributed
computation by anonymous, identical, and mobile finite-state agents. Initially
introduced to model networks of passively mobile sensors, they also capture 
the essence of distributed computation in trust propagation or chemical reactions, the latter
under the name of chemical reaction networks (see \eg~\cite{SCWB08}). 
Structurally, population protocols
can also be seen as a special class of Petri nets or vector addition
systems
\cite{EGLM17}.

Since the agents executing a protocol are anonymous and identical, its
global state---called a \emph{configuration}---is completely
determined by the number of agents at each local state. In each
computation step, a pair of agents, chosen by an adversary subject to
a fairness condition stating that any repeatedly reachable
configuration is eventually reached, interact and move to new states
according to a joint transition function. In a closely related model,
the adversary chooses the pair of agents uniformly at random.

A protocol computes a boolean value for a given initial configuration
if in all fair executions all agents eventually agree to this
value---so, intuitively, population protocols compute by reaching
consensus. Given a set of initial configurations, the predicate
computed by a protocol is the function that assigns to each
configuration $C$ the boolean value computed by the protocol starting
from $C$.

%% If the protocol does not reach consensus for some input, then we
%% say it is ill specified and does not compute any predicate.

Much research on population protocols has focused on their expressive
power, \ie, the class of predicates computable by different classes of
protocols (see e.g. \cite{AACFJP05,AAER07,GR09,MCS11,A17}). 
In a famous result \cite{AAER07}, Angluin et al.\ have shown
that predicates computable by population protocols are exactly the
predicates definable in Presburger arithmetic.
There is also much work on complexity metrics for protocols. The main two metrics
are the \emph{runtime} of a protocol---defined for the model with a randomized adversary as the
expected number of pairwise interactions until all agents have the
correct output value---and its \emph{state space size}, e.g. the number of states of each agent. In~\cite{AAE08a}, Angluin et al.\ show that
every Presburger predicate is computed with high probability by a population protocol with a
leader---a distinguished auxiliary agent that assumes a specific state in the initial configuration irrespective of the input 
--- in $O(n \log^4 n)$ interactions
in expectation, where $n$ is the number of agents of the initial
configuration. Several recent papers study time-space
 trade-offs for specific tasks, like electing a leader~\cite{DS15}, or for specific
predicates, like majority~\cite{AGV15,AAEGR17,BCER17}. 

In this paper we study the state space size of protocols as a function
of the predicate they compute. In particular, we are interested in the
minimal number of states needed to evaluate systems of linear
constraints (a large subclass of the predicates computed by population
protocols) as a function of the number of bits needed to describe the
system.  To the best of our knowledge, this question has not been
considered so far.  We study the question for protocols with and
without leaders. Our results show that
protocols with leaders can be exponentially more compact than
leaderless protocols.

In order to introduce our results in the simplest possible setting, in
the first part of the paper we focus on the family of predicates
$\{x \geq n : n \in \N\}$. These predicates specify the well-known
flock-of-birds problem~\cite{AADFP04}, in which tiny sensors placed on
birds have to reach consensus on whether the number of sick birds in a
flock exceeds a given constant. The minimal number of states for computing
$x \geq n$ formalizes a very natural question about emerging behavior: 
How many states must agents have in order to exhibit a ``phase transition'' 
when their number reaches $n$?  The standard protocol for the
predicate $x \geq n$ (see Example~\ref{ex:flock}) has $n+1$ states. We
are interested in protocols with at most $O(\log n)$ states, either
leaderless or with at most $O(\log n)$ leaders. In the second part of the paper,
we generalize our results to a much larger class of predicates, namely systems
of linear inequalities $A\vec{x} \geq \vec{b}$. Since $x \geq n$
is a (very) special case, our lower bounds for flock-of-birds protocols apply, while
the upper bounds require new (and involved) constructions.

\medskip\noindent \textbf{Protocol size for the flock-of-birds problem.}   
In a first warm-up phase we exhibit a family of leaderless protocols
with only $O(\log n)$ states. More precisely, we prove:
\begin{itemize}
\item[(1)] There exists a family $\{\PP_n : n \in \N\}$ of leaderless population protocols such that $\PP_n$ has
  $O(\log_2 n)$ states and computes the predicate $x \geq n$ for every
  $n \in \N$.
\end{itemize}
We also give a lower bound:
\begin{itemize}
\item[(2)] For every family $\{\PP_n : n \in \N \}$ of leaderless population protocols such that
$\PP_n$ computes $x \geq n$, there exist infinitely many $n$ such that
$\PP_n$ has at least $(\log n)^{1/4}$ states.
\end{itemize}
However, this bound is only \emph{existential} (``there exists
infinitely many $n$'' instead of ``for all $n$'').  Moreover, it
follows from a counting argument that does not provide any information
on the values of $n$ realizing the bound. Is there a poly-logarithmic
universal bound? We show that, surprisingly, the answer is negative:
\begin{itemize}
\item[(3)]There exists a family $\{\PP_n : n \in \N\}$
  of population protocols with two leaders, and values $c_0 <
  c_1 < \ldots \in \N$, such that $\PP_n$ has $O(\log\log c_n)$ states
  and computes the predicate $x \geq c_n$ for every $n \in \N$.
\end{itemize}
Observe that in these protocols the ``phase transition'' occurs at $x=c_n$,
even though no agent has enough memory to
index a particular bit of $c_n$.

Can one go even further, and design $O(\log\log\log c_n)$ protocols?
We show that the answer is negative for \emph{1-aware} protocols. Both
the standard protocol for $x \geq n$
and the families of (1) and (3) have the following, natural
property: If the number of agents is greater than or equal to $n$,
then the agents not only reach consensus 1, they also
eventually \emph{know} that they will reach this consensus. 
We say that these protocols are 1-aware.

We obtain lower bounds for 1-aware protocols that
essentially match the upper bounds of (1) and (3):
\begin{itemize}
\item[(4)] Every leaderless, 1-aware population protocol computing
 $x \geq n$ has at least $\log_3 n$ states.
\item[(5)] Every 1-aware protocol (leaderless or not) computing $x \geq n$ has at least $(\log \log(n) / 151)^{1 / 9}$ states.
\end{itemize}

\medskip\noindent \textbf{Protocols for systems of linear inequalities.} 
In the second part of the paper we show that our results can be
extended to other predicates. First, instead of the simple predicate
$x \geq n$, we study the general linear predicate $a_1 x_1 + a_2 x_2
+ \dots + a_k x_k \geq c$ for arbitrary integer coefficients
$a_1, \ldots, a_k, c \in \Z$. By means of a delicate construction we
give protocols whose number of states grows only logarithmically in
the size of the coefficients:
\begin{itemize}
\item[(6)] There is a protocol with at most $O(kn)$ states and $O(n)$ leaders that computes 
$a_1 x_1+\cdots+a_k x_k\geq c$, where $n$ is the size
of the binary encoding of $\max(|a_1|, |a_2|, \ldots, |a_k|, |c|)$.
\end{itemize}
Finally, in the most involved construction of the paper, we show that
the same applies to arbitrary systems of linear inequalities. Note
that the standard conjunction construction, which produces a protocol
for $\varphi_1 \land \varphi_2$ from protocols computing predicates
$\varphi_1$ and $\varphi_2$, cannot be applied because it would lead
to exponentially large protocols.
\begin{itemize}
\item[(7)] There is a protocol with at most $O((\log m + n)(m + k))$ states and $O(m(\log m + n))$ leaders 
that computes
$A\vec{x} \geq \vec{c}$, where $A \in \Z^{m \times k}$ and $n$ is the
size of the largest entry in $A$ and $\vec{c}$.
\end{itemize}

%\medskip\noindent \textbf{Related work.} The minimal number of states
%has been determined for single predicates, like the majority
%predicate $x \geq y$, which is known to require at least 4 states
%in~\cite{MNRS17}. Our results are for families of predicates. There is
%also recent work on time-space trade-offs of the speed of a population
%protocol and its number of states (see, \eg, \cite{AAEGR17} and the
%references therein). However, these results are achieved for a
%different model in which the number of states of an agent is not
%constant, but a function of the number of agents. In other words,
%while in the standard model a protocol with a fixed number of states
%computes the predicate for all possible inputs, in this model the
%predicate is computed by an infinite family of protocols, one for each
%possible input.

\medskip\noindent \textbf{Structure of the paper.} Section~\ref{sec:preliminaries}
introduces basic definitions, protocols with and without leaders, and
a simple construction with an involved correctness proof showing how
to simulate protocols with $k$-way interactions by standard protocols
with binary interactions. Sections~\ref{sec:leaderless}
to \ref{sec:one:aware} present our bounds on the flock-of-birds
predicates, and Section \ref{sec:sys:lin} the bounds on systems of
linear inequalities. Due to space constraints, some proofs are
deferred to the appendix.

%Section~\ref{sec-conclusion} presents some conclusions.

%% showing that $k$-way protocols, i.e., a generalization of the standard
%% mode in which interactions can involve$2, 3, \ldots, k$ processes can
%% be simulated by the standard 2-way protocols with very low blow-up in
%% the number of states.

%% file: sec-preliminaries.tex
\noindent\textbf{Numbers.} Let $n \in \N_{>0}$. The logarithm in base $b$ of $n$ is denoted by $\log_b n$. Whenever $b = 2$, we omit the subscript. We define $\bits{n}$ as the set of indices of the bits occurring in
the binary representation of $n$, \eg $\bits{13} = \{0, 2, 3\}$ since
$13 = 1101_2$. The \emph{size} of $n$, denoted $\size{n}$, is the
number of bits required to represent $n$ in binary. Note that
$|\bits{n}| \leq \size{n} = \lfloor \log n \rfloor + 1$.\medskip

\noindent\textbf{Multisets.} A \emph{multiset} over a finite set $E$
is a mapping $M \colon E \to \N$. The set of all multisets over $E$ is denoted $\N^E$.  
For every $e \in E$, $M(e)$ denotes the number of
occurrences of $e$ in $M$, and for every $E' \subseteq E$
we define $M(E') \defeq \sum_{e \in E'} M(e)$. The \emph{support} and \emph{size} of $M$ are 
defined respectively as $\supp{M} \defeq \{e \in E : M(e) > 0\}$ and $|M| \defeq \sum_{e \in E} M(e)$.  \emph{Addition} and
\emph{comparison} are extended to multisets componentwise, \ie $(M
\mplus M')(e) \defeq M(e) + M'(e)$ for every $e \in E$, and $M \leq M'
\defiff M(e) \leq M(e)$ for every $e \in E$. We define \emph{multiset
  difference} as $(M \mminus M')(e) \defeq \max(M(e) - M'(e), 0)$ for
every $e \in E$. The empty multiset is denoted $\vec{0}$.
We sometimes denote multisets using a
set-like notation, \eg $\multiset{f, 2 \cdot g, h}$ is the multiset
$M$ such that $M(f) = 1$, $M(g) = 2$, $M(h) = 1$ and $M(e) = 0$ for
every $e \in E \setminus \{f, g, h\}$. \medskip

%% file: sec-protocols.tex
\noindent\textbf{Population protocols.} We introduce a rather general
model of population protocols, allowing for interactions between more than 
two agents and for leaders.  A \emph{$k$-way population protocol} is a 
tuple $\PP = (Q, T, I, L, O)$ such that
\begin{itemize}
\item $Q$ is a finite set of \emph{states},
\item $T \subseteq \bigcup_{2 \leq i \leq k} Q^i \times Q^i$ is a set of 
 \emph{transitions},
\item $I \subseteq Q$ is a set of \emph{initial states},
\item $L \in \N^Q$ is a set of \emph{leaders}, and 
\item $O \colon Q \to \{0, 1\}$ is the \emph{output mapping}.
\end{itemize}
We assume throughout the paper that agents can always
interact, \ie, that for every pair of states $(p, q)$, there exists a
pair of states $(p', q')$ such that $((p, q), (p', q')) \in T$.

A \emph{configuration} of $\PP$ is a multiset $C \in \N^Q$ such that
$|C| > 0$. Intuitively, $C$ describes a non empty collection
containing $C(q)$ agents in state $q$ for every $q \in Q$. We denote
the set of configurations over $E \subseteq Q$ by $\pop{E}$. A
configuration $C$ is \emph{initial} if $C = D \mplus L$ for some
$D \in \pop{I}$. So, intuitively, leaders are distinguished agents
that are present in every initial configuration. The \emph{number of
leaders} of $\PP$ is $|L|$.  We say that $\PP$ is \emph{leaderless} if
it has no leader, \ie if $L = \vec{0}$. We discuss protocols with and
without leaders later in this section.

Let $t = ((p_1, p_2, \ldots, p_i), (q_1, q_2, \ldots, q_i))$ be a
transition. To simplify the notation, we denote $t$ as $p_1,
p_2, \ldots, p_i \mapsto q_1, q_2, \ldots, q_i$. Intuitively, $t$ describes
that $i$ agents at states $p_1, \ldots, p_i$ may interact and move to states
$q_1, \ldots, q_i$. The \emph{preset} and
\emph{postset} of $t$ are respectively defined as $\pre{t} \defeq
\{p_1, p_2, \ldots, p_i\}$ and $\post{t} \defeq \{q_1, q_2, \ldots,
q_i\}$. We extend presets and postsets to sets of transitions, \eg
$\pre{T} \defeq \bigcup_{t \in T} \pre{t}$. The \emph{pre-multiset}
and \emph{post-multiset} of $t$ are respectively defined as
$\prem{t} \defeq \multiset{p_1, p_2, \ldots, p_i}$ and
$\postm{t} \defeq \multiset{q_1, q_2, \ldots, q_i}$.

We say that $t$ is \emph{enabled} at $C \in \pop{Q}$ if $C
\geq \prem{t}$. If $t$ is enabled at $C$, then it can \emph{occur}, in
which case it leads to the configuration $C' = (C \mminus \prem{t})
\mplus \postm{t})$. We denote this by $C \trans{t} C'$. 
We say that $t$ is \emph{silent} if $\prem{t} = \postm{t}$. In particular,
if $t$ is silent and $C \trans{t} C'$, then $C=C'$. We write $C
\trans{} C'$ if $C \trans{t} C'$ for some $t \in T$. We write $C
\trans{t_1 t_2 \cdots t_k} C'$ if there exist $C_0, C_1, \ldots, C_k
\in \pop{Q}$ and $t_1, t_2, \ldots, t_k \in T$ such that $C = C_0
\trans{t_1} C_1 \trans{t_2} \cdots C_k = C'$. We write $C \trans{*}
C'$ if $C \trans{\sigma} C'$ for some $\sigma \in T^*$. We say that
$C'$ is \emph{reachable} from $C$ if $C \trans{*} C'$. The
\emph{support} of a sequence $\sigma = t_1 t_2 \cdots t_n \in T^*$ is
$\supp{\sigma} \defeq \{t_i : 1 \leq i \leq n\}$. 

\begin{example}\label{ex:flock}
  The flock-of-birds protocol mentioned in the introduction is
  formally defined as $\PP_n = (Q, T, I, L, O)$ where $Q = \{0, 1,
  \ldots, n\}$, $I = \{1\}$, $L = \vec{0}$, $O(a) = 1 \iff a = n$, and where
  $T$ consists of the following transitions:
  \begin{align*}
    s_{a,b}:\ & a, b \mapsto 0, \min(a + b, n) &&\text{for every } 0 \leq
    a, b < n, \\
    t_a:\ & a, n \mapsto n, n &&\text{for every } 0 \leq a \leq n.
  \end{align*}
  $\PP_n$ is 2-way and leaderless. Intuitively, it works as follows. Each agent stores a number. When two
  agents meet, one agent stores the sum of their values and the other
  one stores 0. Sums cap at $n$. Once an agent reaches $n$, all agents
  eventually get converted to $n$. To illustrate the above
  definitions, observe that: $\pre{s_{2,3}} = \{2, 3\}$, $\post{t_2} =
  \{n\}$, $\prem{s_{2,3}} = \multiset{2, 3}$ and $\postm{t_{2}} =
  \multiset{n, n}$. Configuration $\multiset{1, 1, 1}$ is initial, but
  $\multiset{1, 0, 2}$ is not. We have $\multiset{1,
    1, 1} \trans{s_{1,1}} \multiset{1, 0, 2} \trans{t_0} \multiset{1,
    2, 2} \trans{t_1} \multiset{2, 2, 2}$, or more concisely
  $\multiset{1, 1, 1} \trans{\sigma} \multiset{2, 2, 2}$ where $\sigma
  = s_{1,1} t_0 t_1$. \qed
\end{example}

\noindent\textbf{Computing with population protocols.} An \emph{execution} 
$\pi$ is an infinite sequence of configurations
$C_0 C_1 \cdots$ such that $C_0 \trans{} C_1 \trans{} \cdots$. We say
that $\pi$ is \emph{fair} if for every configuration $D$ the following
holds\footnote{This definition of fairness differs from the original
  definition of Angluin et al.~\cite{AADFP04}, but is equivalent.}:
\begin{align*}
  \text{if } \{i \in \N : C_i \trans{*} D\} \text{ is infinite, then }
  \{i \in \N : C_i = D\} \text{ is infinite}.
\end{align*}
In other words, fairness ensures that a configuration cannot be
avoided forever if it can be reached infinitely often along $\pi$. We
say that a configuration $C$ is a \emph{consensus configuration} if
$O(p) = O(q)$ for every $p, q \in \supp{C}$. If a configuration $C$ is
a consensus configuration, then its \emph{output} $O(C)$ is the unique
output of its states, otherwise it is $\bot$. An execution $\pi = C_0
C_1 \cdots$ \emph{stabilizes} to $b \in \{0, 1\}$ if $O(C_i) =
O(C_{i+1}) = \cdots = b$ for some $i \in \N$. The \emph{output} of
$\pi$ is $O(\pi) \defeq b$ if it stabilizes to $b$, and $O(\pi) \defeq
\bot$ otherwise. A consensus configuration $C$ is \emph{stable} if
every configuration $C'$ reachable from $C$ is a consensus
configuration such that $O(C') = O(C)$. It can easily be shown that a
fair execution stabilizes to $b \in \{0, 1\}$ if and only if it
contains a stable configuration whose output is $b$.

A population protocol $\PP = (Q, T, I, L, O)$ is \emph{well-specified}
if for every initial configuration $C_0$, there exists $b \in \{0,
1\}$ such that every fair execution $\pi$ starting at $C_0$ has output
$b$. If $\PP$ is well-specified, then we say that it \emph{computes}
the predicate $\varphi \colon \pop{I} \to \{0, 1\}$ if for every
$D \in \pop{I}$, every fair execution starting at $D \mplus L$ has
output $\varphi(D)$.

\begin{example}
  Consider the protocol $\PP_2$ defined in Example \ref{ex:flock} (i.e, $n=2$). We have $O(\multiset{1, 1,
    1}) = 0$, $O(\multiset{2, 2, 2}) = 1$ and $O(\multiset{1, 0, 2}) =
  \bot$. The execution $\multiset{1, 1, 1} \trans{} \multiset{1, 0, 2}
  \trans{} \multiset{1, 2, 2} \trans{} \multiset{2, 2 ,2} \trans{}
  \multiset{2, 2 ,2} \trans{} \cdots$ is fair and its output is
  $1$. However, the execution $\multiset{1, 1, 1} \trans{}
  \multiset{1, 0, 2} \trans{} \multiset{1, 0, 2} \trans{} \cdots$ is
  not fair since $\multiset{1, 0, 2}$ occurs infinitely often and can
  lead to $\multiset{2, 2, 2}$ which does not occur.
\end{example}

\noindent\textbf{Leaders.}  Intuitively, leaders are extra agents present in every
initial configuration. Allowing a large number of leaders may help to
compute predicates with fewer states. To illustrate this, consider 
the leaderless protocol of Example \ref{ex:flock}. It computes  $x \geq n$
with $n+1$ states. We describe a 2-way protocol with only 4 states, but $n$ leaders. 
It is an adaptation of the well-known basic majority protocol
(see, \eg, \cite{AR07}). Let $\PP_n' = (Q, T, I, L_n, O)$ be the protocol where 
$Q \defeq \{x, y, \overline{x},
\overline{y}\}$, $I \defeq \{x\}$, $L_n \defeq \multiset{n \cdot y}$,
$O(x) = O(\overline{x}) \defeq 1$, $O(y) = O(\overline{y}) \defeq 0$,
and where $T$ consists of the following transitions:
\begin{align*}
  x, y &\mapsto \overline{x}, \overline{y}, &
  x, \overline{y} &\mapsto x, \overline{x}, &
  y, \overline{x} &\mapsto y, \overline{y}, &
  \overline{x}, \overline{y} &\mapsto \overline{x}, \overline{x}.
\end{align*}

Informally, ``active'' agents in states $x$ and $y$ collide and become
``passive'' agents in states $\overline{x}$ and $\overline{y}$. At
some point, some active agents ``win'' and convert all passive agents
to their output. It is known that this protocol is well-specified and
computes the predicate $x \geq y$ when there are no leaders (\ie, if
we set $L_n = \vec{0}$). So, by initially fixing $n$ leaders in state
$y$, $\PP_n'$ computes $x \geq n$.

Thus, the predicate $x \geq n$ can be computed either with $O(n)$ states
and no leaders, or with $4$ states and $O(n)$ leaders. This indicates
a trade-off between states and leaders, and one should avoid hiding
all of the complexity in one of them. For this reason, we make these
two quantities explicit in all of our results.

The reason for considering protocols with leaders is that, as we shall
see, even a constant number of leaders demonstrably leads to
exponentially more compact protocols for some predicates.  Other
papers have made similar observations with respect to other resource
measures (see \eg~\cite{AAE08a,LFIISV17}).

%% file: sec-k-way.tex
\medskip\noindent\textbf{From $k$-way to $2$-way protocols.} In our constructions it is very convenient to use  $k$-way transitions for $k > 2$.
The following lemma shows that $k$-way protocols can be transformed
into $2$-way protocols by introducing a few extra states. Intuitively,
a $k$-way transition is simulated by a chain of 2-way transitions.
The first part of the chain ``collects'' $k$ participants one by
one. First, two agents agree to participate, and one of them becomes
``passive'', while the second ``searches'' for a third
participant. This is iterated until $k$ participants are collected.
In the second part, the last collected agent ``informs'' all passive
agents, one by one, that $k$ agents have been collected; upon hearing
this, the passive agents move to their destination states and become
active again. To prevent faulty behavior
when there are not enough agents, all transitions of the first part
can be ``reversed'', that is, the agent that is currently searching
and the last collected agent can ``repent'' and ``undo'' the
transition. While the construction is simple and intuitive, its
correctness proof is very involved, because agents that reach their
destination can engage in other interactions while other participants
are still passive. The construction and the correctness proof are
presented in Appendix~\ref{app:k:way}.

\begin{restatable}{restatelemma}{lemKway}\label{lem:kway}
  Let $\PP = (Q, T, I, L, O)$ be a well-specified $k$-way population
  protocol. For every $3 \leq i \leq k$, let $n_i$ be the number of
  $i$-way transitions of $\PP$. There exists a 2-way population
  protocol $\PP'$, with at most $|Q| + \sum_{3 \leq i \leq k} 3i \cdot
  n_i$ states, which is well-specified and computes the same predicate
  as $\PP$.
\end{restatable}

%% file: sec-leaderless.tex
In this section, we consider \emph{leaderless} protocols for the
predicate $x \geq n$. We first show that the number of states required
to compute this predicate can be reduced from the known $O(n)$ bound
to $O(\log n)$, using a similar binary encoding as in ~\cite{AAEGR17}.
Then we show an existential lower bound of $O((\log
n)^{1/4})$.\medskip

%\subsection{A protocol with $O(\log n)$ states}\label{ssec:leaderless:logn}

\noindent\textbf{A protocol with $O(\log n)$ states.} We describe a leaderless $\size{n}$-way protocol $\PP_n = (Q_n, T_n,
I_n, \vec{0}, O_n)$ with $\size{n} + 3$ states that computes $x \geq
n$. The states are
$Q_n \defeq \{\state{0}, \state{2^0}, \ldots, \state{2^{\size{n}}},
\state{n}\}$ and the sole initial state is $I_n \defeq \{\state{2^0}\}$. The output mapping is defined as
$O_n(\state{n}) \defeq 1$ and $O_n(q) \defeq 0$ for every state
$q \neq \state{n}$.
 
Before defining the set $T_n$ of transitions, we need some
preliminaries. For every state $q \in Q_n$, let $\val(q)$ denote the
number $q$ stands for, \ie $\val(\state{0}) = 0$, $\val(\state{n}) =
n$ and $\val(\state{2^i}) = 2^i$ for every $0 \leq
i \leq \size{n}$. Moreover, for every configuration $C$, let
$\val(C) \defeq \sum_{q \in Q_n} \val(q) \cdot C(q)$. A configuration
$C$ is a \textit{representation of $m$} if $\val(C) = m$. For
example, the configuration $\multiset{\state{0}, \state{2^1},
5 \cdot \state{2^3}}$ is a representation of $0 + 2^1 + 5 \cdot 2^3 =
42$. Observe that every initial configuration $C_0$ is a
representation of $|C_0|$.
 
$T_n$ is the union of two sets $T_n^1$ and $T_n^2$. Intuitively,
$T_n^1$ allows the protocol to reach from a representation of a
number, say $m$, other representations of $m$. Formally, the
transitions of $T_n^1$ are:
$$
\begin{array}{rcll}
  \state{2^i}, \state{2^i} & \mapsto & \state{2^{i+1}}, \state{0} &
  \mbox{for every $0 \leq i < \size{n}$} \\

  \state{2^{i+1}}, \state{0} & \mapsto & \state{2^i}, \state{2^i} &
  \mbox{for every $0 \leq i < \size{n}$} \\

  \multiset{\state{2^i} : i \in \bits{n}} & \mapsto & \state{n},
  \underbrace{\state{0}, \cdots, \state{0} }_{\mathclap{|\bits{n}|-1
      \text{ copies}}}
\end{array}$$
The transitions of $T_n^2$ allow agents in state $\state{n}$ to
``attract'' all other agents to $\state{n}$. Formally, they are:
$$
\begin{array}{rcll}
  \state{n}, q & \mapsto & \state{n}, \state{n} & \mbox{for every
    $q \in Q_n$}
\end{array}
$$

Let us show that $\PP_n$ computes $x \geq n$. Let $C_0 = \multiset{m
  \cdot \state{2^0}}$. If $m < n$, then $C(\state{n}) = 0$ holds for
every representation $C$ of $m$. Therefore, every configuration $C$
reachable from $C_0$ satisfies $C(\state{n}) = 0$ and, since
$\state{n}$ is the only state with output 1, the protocol stabilizes
to $0$. If $m \geq n$, then it is possible to reach a representation
$C$ of $m$ satisfying $C(\state{n}) > 0$, for example $C =
\multiset{\state{n}, (m-n) \cdot \state{2^0}}$. Since for every
transition $\state{2^i}, \state{2^i} \mapsto \state{2^{i+1}},
\state{0}$ the set $T_n$ also contains the reverse transition
$\state{2^{i+1}}, \state{0} \mapsto \state{2^i}, \state{2^i}$, every
representation $C$ of $m$ satisfying $C(\state{n}) = 0$ can reach a
representation $C'$ of $m$ satisfying $C'(\state{n}) > 0$. Let $\pi = C_0 C_1
C_2 \cdots$ be a fair execution. By fairness, there is some $i
\in \N$ such that $C_i(\state{n}) > 0$. Again by fairness, and because
of $T_n^2$, there is also an index $j$ such that $C_k = \multiset{m
  \cdot \state{n}}$ for every $k \geq j$, and so $\pi$ stabilizes to
1.

Note that $|Q_n| = \size{n} + 3$. Moreover, $\PP_n$ has one
$|\bits{n}|$-way transition. Thus, by Lemma~\ref{lem:kway}, we obtain
the following theorem:

\begin{theorem}
 There exists a family $\{\PP_0, \PP_1, \ldots\}$ of leaderless and
 2-way population protocols such that $\PP_n$ has at most $4\lfloor
 \log n \rfloor + 7$ states and computes the predicate $x \geq n$.
\end{theorem}

%% \subsection{An existential $(\log n)^{1/4}$ lower bound}\label{ssec:log:lower}

\noindent\textbf{An existential $(\log n)^{1/4}$ lower bound.} We show that every family $\{\PP_n\}_{n \in \N}$ of leaderless and
2-way protocols computing the family of predicates $\{x \geq
n\}_{n \in \N}$ must contain infinitely many members of size
$\Omega((\log n)^{1/4})$. We call this an existential lower bound,
contrary to a universal lower bound, which would state that $\PP_n$
has size $\Omega((\log n)^{1/4})$ for every $n \geq 1$.

\begin{restatable}{theorem}{thmExistBound}\label{thm:existbound}
  Let $\{\PP_0, \PP_1, \ldots \}$ be an infinite family of leaderless
  and 2-way population protocols such that $\PP_n$ computes the
  predicate $x \geq n$ for every $n \in \N$. There exist infinitely
  many indices $n$ such that $\PP_n$ has at least $(\log n)^{1/4}$
  states.
\end{restatable}

\begin{proof}[Proof sketch.]
  The proof boils down to bounding the number $d(m)$ of unary
  predicates computed by protocols with $m$ states. The number of
  distinct sets of transitions, excluding silent ones, is bounded by
  $2^{m^4 - m^2}$. The number of possible initial states and output
  mappings are respectively $m$ and $2^m$. Altogether, we obtain:
 \begin{align*}
   d(m) &\leq 2^{m^4 - m^2} \cdot m \cdot 2^m  = 2^{m^4}
   \cdot \frac{2^m \cdot m}{2^{m^2}}  \leq
   2^{m^4}. \qedhere
 \end{align*}
\end{proof}

%% file: sec-loglog-upper.tex
%% Local macros
\newcommand{\semig}[1]{{\cal S}_{#1}}
\newcommand{\al}{{\cal A}}

The lower bound of Section~\ref{sec:leaderless} is not valid for every
$n$, it only ensures that, for some values of $n$, protocols computing
$x \geq n$ must have a logarithmic number of states. We prove that,
surprisingly, there is an infinite sequence $n_1 < n_2 < \cdots$ of
values that break through the logarithmic barrier: The predicates $x
\geq n_i$ can be computed by very small protocols with only $O( \log
\log n_i)$ states and two leaders. So, loosely speaking, a flock of
birds can decide if it contains at least $n_i$ birds, even though no
bird has enough memory to store even one single bit of $n_i$.

The result is based on a construction of~\cite{MM82}. In this paper,
Mayr and Meyer study the word problem for commutative semigroup
presentations. Given a finite set $\al$ of generators, a presentation
of a commutative semigroup generated by $\al$ is a finite set of
productions $\semig{} = \{l_1 \rightarrow r_1, \ldots, l_m \rightarrow
r_m\}$, where $l_i, r_i \in \al^*$ for every $1 \leq i \leq m$,
satisfying:
\begin{itemize}
\item Commutativity: $ab \rightarrow ba \in \semig{}$ for every $a,b
  \in \al$;\footnote{In~\cite{MM82}, the elements of $S$ are written
    using uppercase letters. We use lowercase for convenience.} and
\item Reversibility: if $l \rightarrow r \in \semig{}$, then $r
  \rightarrow l \in \semig{}$.
\end{itemize}
Given $\alpha, \beta \in \al^*$, we say that $\beta$ is \emph{derived}
from $\alpha$ \emph{in one step}, denoted by $\alpha \trans{} \beta$,
if $\alpha = \gamma \, l \, \delta$ and $\beta = \gamma \, r \,
\delta$ for some $\gamma, \delta \in \al^*$ and some $r \rightarrow l
\in \semig{}$. We say that $\beta$ is \emph{derived} from $\alpha$ if
$\alpha \trans{*} \beta$, where $\trans{*}$ is the reflexive
transitive closure of the relation induced by $\trans{}$. Observe
that, by reversibility, we have $\alpha \trans{*} \beta$ if{}f $\beta
\trans{*} \alpha$. Further, by commutativity we have $\alpha \trans{*}
\beta$ if{}f $\pi(\alpha) \trans{*} \pi'(\beta)$ for every permutation
$\pi$ of $\al$.

Mayr and Meyer study the following question: given a commutative
semigroup presentation $\semig{}$ over $\al$, and initial and final
letters $s,f \in \al$, what is the length of the shortest word
$\alpha$ such that $s \trans{*} f \alpha$? They exhibit a family of
presentations of size $O(n)$ for which the shortest $\alpha$ has
double exponential length $2^{2^n}$. More precisely,
in~\cite[Sect.~6]{MM82}, they construct a family $\{\semig{n}\}_{n
  \geq 1}$ of presentations over alphabets $\{\al_n\}_{n \geq 1}$
satisfying the following properties:
\begin{itemize}
\item[(1)] $|\al_n| = 14n + 10$, $|\semig{n}| = 20n + 8$, and
  $\max\{|l|, |r| : l \rightarrow r \in \semig{n} \} = 5$.
\item[(2)] $\{s_n,f_n,b_n, c_n\} \subseteq \al_n$ for every $n \geq
  1$.
\item[(3)] $s_n c_n \trans{*} f_n\alpha$ if{}f $\alpha =
  c_nb_n^{2^{2^n}}$~\cite[Lemma~6 and~8]{MM82}.
\end{itemize}
\noindent To apply this result, for each $n \geq 1$ we construct a
$5$-way population protocol $\PP_n = (Q_n, T_n, I_n, L_n, O_n)$ with
two leaders as follows:
\begin{itemize}
\item $Q_n \defeq \al_n \cup \{x\}$ for some $x \notin \al_n$.
\item $T_n \defeq T_n^1 \cup T_n^2$, where:
\begin{itemize}
\item $T_n^1$ contains a transition $\mathrm{pad}(p)$ for every
  production $p= l \rightarrow r$ of $\semig{n}$, obtained by
  ``padding'' $p$ with $x$ so that its left and right sides have the
  same length. For example, $\mathrm{pad}(aab \rightarrow cd) = a, a,
  b \mapsto c, d, x$, and $\mathrm{pad}(a \rightarrow bc) = a, x
  \mapsto b, c$,
\item $T_n^2 \defeq \{f_n, q \mapsto f_n, f_n \mid q \in Q_n\}$,
\end{itemize}
\item $I_n \defeq \{x\}$, 
\item $L_n \defeq \multiset{c_n, s_n}$, and
\item $O_n(f_n) \defeq 1$ and $O_n(q) \defeq 0$ for every $q \neq
  f_n$.
\end{itemize}
Intuitively, $T_n^1$ allows $\PP_n$ to simulate derivations of
$\semig{n}$: a step $C \trans{\mathrm{pad}(p)} C'$ of $\PP_n$
simulates a one-step derivation of $\semig{n}$. We make this more
precise. Given $\alpha \in \al_n^*$ and $m \geq |\alpha|$, let
$C_{\alpha, m}$ be the configuration of $\PP_n$ defined as follows:
$C_{\alpha, m}(x) = m$, and $C_{\alpha, m}(a) = |\alpha|_a$ for every
$a \in \al_n$, where $|\alpha|_a$ is the number of occurrences of $a$
in $\alpha$. Further, given a configuration $C$ of $\PP_n$, let
$\alpha_C$ be the element of $\semig{n}$ given by $\alpha_C =
a_1^{C(a_1)} \cdots a_m^{C(a_m)}$, where $a_1, \ldots, a_m$ is a fixed
enumeration of $\al_n$. We have:

\begin{restatable}{restatelemma}{lemSemigroupSim}\label{lem:semigroup:sim}
  Let $\alpha, \beta \in \al_n^*$ and let $C, C'$ be configurations of
  $\PP_n$.
  \begin{itemize}
  \item[(a)] If $\alpha \trans{p_1 \cdots p_k} \beta$ in $\semig{n}$,
    then for every $m \geq 4k$, $C_{\alpha, m}
    \trans{\mathrm{pad}(p_1) \cdots \mathrm{pad}(p_k)} C_{\beta, m'}$
    in $\PP_n$ for some $m' \geq 0$.
  \item[(b)] If $C \trans{\mathrm{pad}(p_1) \cdots \mathrm{pad}(p_k)}
    C'$ in $\PP_n$, then $\alpha_C \trans{p_1 \cdots p_k} \alpha_{C'}$
    in $\semig{n}$.
  \end{itemize}
\end{restatable}

From Lemma~\ref{lem:semigroup:sim}, (1) and~(3), the following can be
shown:

\begin{restatable}{theorem}{thmLoglogUpperbound}\label{thm:loglog:upperbound}
  For every $n \in \N$, there is a 5-way protocol $\PP_n$ with at most
  $14n + 11$ states and at most $34n + 19$ transitions that computes
  the predicate $x \geq c_n$ for some number $c_n \geq 2^{2^n}$.
\end{restatable}

%% \begin{proof}[Proof sketch.]
%%   Let $C_0$ be an initial configuration. If $f_n$ is coverable from
%%   $C_0$ then, by the reversibility property of $\semig{n}$, $f_n$ is
%%   coverable from every configuration reachable from $C_0$. Thus, every
%%   fair execution from $C_0$ stabilizes to 1. If $f_n$ is not coverable
%%   from $C_0$, then we must have $O(C_0) = 0$ as $f_n$ is the only
%%   state with output 1.

%%   By~(3) and Lemma~\ref{lem:semigroup:sim}, there exists an initial
%%   configuration $C_0$ such that $O(C_0) = 1$. Let $C_0$ be the
%%   smallest such configuration. We have $|C_0| \geq 2^{2^n}$. Moreover,
%%   since $f_n$ is coverable from every configuration $C_0' \geq C_0$,
%%   the protocol computes the predicate $x \geq |C_0|$.
%% \end{proof}

Using Theorem~\ref{thm:loglog:upperbound} and Lemma~\ref{lem:kway}, we obtain:

\begin{restatable}{restatecorollary}{corLoglogUpperbound}\label{cor:loglog:upperbound}
  There exists a family $\{\PP_0, \PP_1, \ldots\}$ of 2-way protocols
  with two leaders and a family $\{c_0, c_1, \ldots\}$ of natural
  numbers such that for every $n \in \N$ the following holds: $c_n
  \geq 2^{2^n}$ and protocol $\PP_n$ has at most $314 \log\log c_n +
  131$ states and computes the predicate $x \geq c_n$.
\end{restatable}

%% file: sec-one-aware.tex
To the best of our knowledge, all the protocols in the literature for
predicates $x \geq n$, including those of Section~\ref{sec:leaderless}
and Section~\ref{sec:loglog:upper}, share a very natural property: if
the number of agents is greater than or equal to $n$, then the agents
not only eventually reach consensus 1, they also eventually
\emph{know} that they will reach this consensus. Let us formalize this
idea:

\begin{definition}
\label{def:1aware}
  A well-specified population protocol $\PP = (Q, T, I, L, O)$ is
  \emph{1-aware} if there is a set $Q_1 \subseteq Q \setminus (I \cup
  \supp{L})$ of states such that for every initial configuration $C_0$
  and every fair execution $\pi = C_0 C_1 \cdots$
  \begin{itemize}
  \item[(1)] if $\pi$ stabilizes to $0$, then $C_i(Q_1) = 0$ for every
    $i \geq 0$, and
  \item[(2)] if $\pi$ stabilizes to $1$, then there is some $i \geq 0$
    such that $C_j(Q \setminus Q_1) = 0$ for every $j \geq i$.
  \end{itemize}
\end{definition}

If in the course of an execution $\pi$ an agent reaches a state of
$Q_1$, then $\pi$ cannot stabilize to 0 by~(1), and so, since $\PP$ is
well-specified, it stabilizes to 1; intuitively, at this moment the
agent ``knows'' that the consensus will be 1. Further, if an execution
stabilizes to 1, then all agents eventually reach and remain in $Q_1$
by~(2), and so eventually all agents ``know''.\footnote{We could also
require the seemingly weaker property that eventually at least one
agent ``knows''. However, by adding transitions that ``attract'' all
other agents to $Q_1$, we can transform a protocol in which some agent
``knows'' into a protocol computing the same predicate in which all
agents ``know''.}  Albeit seemingly restrictive, 1-aware protocols
compute a significant subclass of predicates: monotonic Presburger
predicates (see Appendix~\ref{app:monotonic} for more details).

We say that a state $q$ is \emph{coverable} from a configuration $C$
if $C \trans{*} C'$ for some configuration $C'$ such that $C'(q) >
0$. The fundamental property of 1-aware protocols is that, loosely
speaking, consensus reduces to coverability:

\begin{restatable}{restatelemma}{lemSimpleOneAware}\label{lem:simple1aware}
  Let $\PP = (Q, T, \{x\}, L, O)$ be a 1-aware protocol computing a
  unary predicate $\varphi$. We have $\varphi(n) = 1$ if and only if
  some state of $Q_1$ is coverable from $\multiset{n \cdot x}+L$.
\end{restatable}

We show that for 1-aware protocols, the bounds of
Sections~\ref{sec:leaderless} and~\ref{sec:loglog:upper} are
essentially tight. \medskip

%In Section~\ref{subsec:leaderless}, we prove that
%every 1-aware, leaderless and 2-way protocol computing $x \geq n$ has
%at least $\log_3 n$ states. In Section~\ref{subsec:withleaders}, we
%prove that every 1-aware and 2-way protocol computing $x \geq n$ has
%$\Omega(\log \log_n)$ states.

%% We assume that the reader is familiar with basic definitions of
%% Petri nets (otherwise, the main definitions are given in the
%% Appendix). Given a population protocol $\PP=(Q,T,I,O)$, we
%% associate to it a weighted Petri net $N=(P,T,W)$ where $P = Q$, for
%% every transition $\tau =p_1, \ldots, p_k \mapsto q_1, \ldots, q_k$
%% the set $T$ contains a transition $t_\tau$ with $\preset{t_\tau} =

% \subsection{Leaderless protocols}\label{subsec:leaderless}

\noindent\textbf{Leaderless protocols.} We prove that a 1-aware, leaderless and 2-way protocol computing
$x \geq n$ has at least $\log_3 n$ states. By
Lemma \ref{lem:simple1aware}, it suffices to show that some state of
$Q_1$ is coverable from $\multiset{3^k \cdot q}$, where $q$ is the
initial state. Proposition~\ref{prop:short:saturation} below is the
key to the proof. It states that for every finite execution
$C_1 \trans{\pi} C_2$, there is $C_1' \trans{\pi'} C_2'$ such that
$C_1'$ has the same support as $C_1$ and is not too large, and $C_2'$
contains a ``record'' of all states encountered during the execution
of $\pi$ (this is the set $\supp{C_1} \cup \post{\supp{\pi}}$).

Let us define the \emph{norm} of a configuration $C$ as
$\norm{C} \defeq \max\{C(q) : q \in \supp{C}\}$. We obtain:

\begin{restatable}{restateproposition}{propShortSaturation}\label{prop:short:saturation}
  Let $\PP = (Q, T, I, L, O)$ be a $k$-way population protocol and let
  $C_1 \trans{\pi} C_2$ be a finite execution of $\PP$. There exists a
  finite execution $C_1' \trans{\pi'} C_2'$ such that (a) $\supp{C_1'}
  = \supp{C_1}$, (b) $\supp{C_2'}
  = \supp{C_1} \cup \post{\supp{\pi'}}$, and (c) $\norm{C_1'} \leq
  (k+1)^{|Q|}$.
\end{restatable}

Proposition~\ref{prop:short:saturation} leads to:

\begin{theorem}\label{thm:1awareleaderless}
  Every 1-aware, leaderless and 2-way population protocol $\PP = (Q,
  T, \{q_0\}, \vec{0}, O)$ computing $x \geq n$ has at least $\log_3
  n$ states.
\end{theorem}

\begin{proof}
  Let $Q_1 \subseteq Q$ be the set of states from the definition of
  1-awareness. Since $L = \vec{0}$, $C_0 = \multiset{n \cdot q_0}$ is
  the smallest initial configuration with output 1, and by
  Lemma~\ref{lem:simple1aware} the smallest initial configuration from
  which some state $q_1 \in Q_1$ is coverable. Let $C_0 \trans{\pi}
  C \geq \multiset{q_1}$. Since $q_1 \neq q_0$, we have
  $q_1 \in \post{\supp{\pi}}$. By
  Proposition \ref{prop:short:saturation}, and since $\PP$ is 2-way,
  $q_1$ is also coverable from $C_0'$ satisfying $\supp{C_0'}
  = \supp{C_0} = \{q_0\}$ and $\norm{C_0'} = 3^{|Q|}$. Thus, $C_0'
  = \multiset{3^{|Q|} \cdot q_0}$. By minimality of $n$, we get
  $n \leq 3^{|Q|}$, and thus $|Q| \geq \log_3 n$.
\end{proof}

Observe that the proof Theorem~\ref{thm:1awareleaderless} uses the
fact that $\PP$ is leaderless to conclude $C_0' = \multiset{3^{|Q|}
  \cdot q_0}$ from $\supp{C_0'} = \supp{C_0}$ and $\norm{C_0'} =
3^{|Q|}$, which is not necessarily true with leaders.\medskip

%% \subsection{Protocols with leaders}\label{subsec:withleaders}

\noindent\textbf{Protocols with leaders.} In the case of protocols with leaders we obtain a lower bound from
Rackoff's procedure for the coverability problem of vector addition
systems~\cite{Rac78}.

A \emph{vector addition system} of dimension $k$ ($k$-VAS) is a pair
$(A, \vec{v}_0)$, where $\vec{v}_0 \in \N^k$ is an initial vector and
$A \subseteq \Z^k$ is a set of vectors. An execution of a $k$-VAS is a
sequence $\vec{v}_0 \vec{v}_1 \cdots \vec{v}_n$ of vectors of $\N^k$
such that each $\vec{v}_{i+1} = \vec{v}_i + \vec{a}_i$ for some
$\vec{a}_i \in A$. We write $\vec{v}_0 \trans{*} \vec{v}_n$ and say
that the execution has \emph{length} $n$. A vector $\vec{v}$ is
\emph{coverable} in $(A, \vec{v}_0)$ if $\vec{v}_0 \trans{*} \vec{v}'$
for some $\vec{v}' \geq \vec{v}$. The \emph{size} of a vector $\vec{v}
\in \Z^k$ is $\sum_{1 \leq i \leq k} \size{\max(|\vec{v}(i)|,
  1)}$. The \emph{size} of a set of vectors is the sum of the size of
its vectors. In \cite{Rac78} Rackoff proves:

\begin{theorem}[{\cite{Rac78}}]\label{thm:rackoff}
  Let $A \subseteq \Z^k$ be a set of vectors of size at most $n$ and
  dimension $k \leq n$, and let $\vec{v}_0 \in \N^k$ be a vector of size
  $n$. For every $\vec{v} \in \N^k$, if $\vec{v}$ is coverable in $(A,
  \vec{v}_0)$, then $\vec{v}$ is coverable by means of an execution of
  length at most $2^{(3n)^n}$.
\end{theorem}

%% Using a standard construction from the Petri net literature, it can
%% can be shown that every $2$-way protocol with $n$ states can be
%% simulated by a VAS of \emph{size} at most $12n^8$.  Applying
%% Theorem~\ref{thm:rackoff} we thus get:

Using a standard construction from the Petri net literature, it can be
 shown that every 2-way protocol $\PP$ with $n$ states can be
 simulated by a VAS $\mathcal{V}_\PP$ of size at most $12n^8$, where
 each execution of $\PP$ has a corresponding execution twice as long
 in $\mathcal{V}_\PP$. Thus, by Theorem~\ref{thm:rackoff}:

\begin{restatable}{restateproposition}{propRackoff}\label{prop:Rackoff}
  Let $\PP = (Q, T, I, L, O)$ be a 2-way population protocol and let
  $q \in Q$. For every configuration $C$, if $q$ is coverable from
  $C$, then it is coverable by means of a finite execution of length
  at most $2^{(3m)^m-1}$ where $m = 12|Q|^8$.
\end{restatable}

Using the above corollary, we derive:

\begin{restatable}{restatetheorem}{thmOneawareleader}\label{thm:1awareleader}
  Let $\PP$ be a 1-aware and 2-way population protocol. For every
  $n \geq 2$, if $\PP$ computes $x \geq n$, then $\PP$ has at least
  $(\log \log(n) / 151)^{1 / 9}$ states.
\end{restatable}

%% file: sec-sys-lin.tex
%% Local macros
\newcommand{\rep}[1]{\mathrm{rep}(#1)}

In Section~\ref{sec:leaderless}, we have shown that the
predicate $x \geq c$ can be computed by a leaderless protocol with
$O(\log c)$ states. In this section, we will see that adding a few
leaders allows to compute systems of linear inequalities. More
formally, we show that there exists a protocol with $O((m + k) \cdot
\log(dm))$ states and $O(m \cdot \log(dm))$ leaders computing the
predicate $A\vec{x} \geq \vec{c}$, where $A \in \Z^{m \times k}$,
$\vec{c} \in \Z^m$ and $d$ is the the largest absolute value occuring
in $A$ and $\vec{c}$.

There are three crucial points that make systems of linear
inequalities more complicated than flock-of-birds predicates: (1)
variables have
\emph{coefficients}, (2) coefficients may be positive or
\emph{negative}, and (3) they are the \emph{conjunction} of linear
inequalities. We will explain how to address the two first points by
considering the special case of linear inequalities. We will then
discuss how to handle the third point.\medskip

%%\subsection{Linear inequalities}

\noindent\textbf{Linear inequalities.} Note that the predicate
$\sum_{1 \leq i \leq k} a_i x_i \geq c$ is equivalent to $\sum_{1 \leq
  i \leq k} a_i x_i + (1-c) > 0$. Therefore, it suffices to describe
protocols for predicates of the form $\sum_{1 \leq i \leq k} a_i x_i +
c > 0$. In order to make the presentation more pleasant, we will first
restrain ourselves to the predicate $ax - by + c > 0$ for some fixed
$a, b \in \N$ and $c \in \Z$. Such a predicate admits the difficult
aspects, \ie coefficients and negative numbers. Moreover, as we will
see, handling more than two variables is not an issue.

Let us now describe a protocol $\PP_\text{lin}$ for the predicate $ax
- by + c > 0$. The idea is to keep a representation of $ax - by + c$
throughout executions of the protocol. Let $n \defeq \size{\max(\log
  |a|, \log |b|, \log |c|, 1)}$. As in Section~\ref{sec:leaderless},
we construct states to represent powers of two. However, this time, we
also need states to represent negative numbers:
\begin{align*}
  Q^+ &\defeq \{\state{+2^i} : 0 \leq i \leq n\}\ \text{ and }\
  Q^-  \defeq \{\state{-2^i} : 0 \leq i \leq n\}.
\end{align*}
We also need states $X \defeq \{\state{x}, \state{y}\}$ for the
variables, and two additional states $R \defeq \{\state{+0},
\state{-0}\}$. The set of all states of $\PP_\text{lin}$ is $Q \defeq
X \cup Q^+ \cup Q^- \cup R$, and the initial states are $I \defeq X$.

Let us explain the purpose of $R$. Intuitively, we would like to have
the transitions:
\begin{align*}
  x &\mapsto \multiset{\state{+2^i} : i \in \bits{a}}\ \text{ and }\
  y  \mapsto \multiset{\state{-2^i} : i \in \bits{|b|}}.
\end{align*}
This way, every agent in state $\state{x}$ (resp.\ $\state{y}$) could
be converted to the binary representation of $a$
(resp.\ $b$). Unfortunately, this is not possible as these transitions
produce more states than they consume. This is where leaders become
useful. If $R$ initially contains enough leaders, then $R$ can act as
a \emph{reservoir} of extra states which allow to ``pad''
transitions. More formally, let $\rep{z} \colon \Z \to \pop{Q \setminus X}$
be defined as follows:
$$
\rep{z} \defeq
\begin{cases}
  \multiset{\state{+2^i} : i \in \bits{z}}   &\text{if } z > 0, \\
  \multiset{\state{-2^i} : i \in \bits{|z|}} &\text{if } z < 0, \\
  \multiset{\state{-0}}                      &\text{if } z = 0.
\end{cases}
$$
For every $r \in R$, we add to $\PP_\text{lin}$ the following transitions:
\begin{align*}
  \mathrm{add}_{\state{x}, r}:\ \state{x}, \underbrace{r, r, \ldots,
    r}_{\mathclap{|\rep{a}|-1 \text{ times}}} &\mapsto \rep{a}\ \text{
    and }\
  \mathrm{add}_{\state{y}, r}:\ \state{y}, \underbrace{r, r, \ldots,
    r}_{\mathclap{|\rep{b}|-1 \text{ times}}} \mapsto \rep{b}.
\end{align*}
We set the leaders to $L \defeq \rep{c} + \multiset{(4n + 2) \cdot
  \state{-0}}$. We claim that $4n + 2$ reservoir states are enough, we
will explain later why. Now, the key idea of the construction is that
it is always possible to put $2n$ agents back into $R$. Thus, fairness
ensures that the number of agents in $X$ eventually decreases to zero,
and then that the value represented over $Q^+ \cup Q^-$ is $ax - by +
c$. We let the representations over $Q^+$ and $Q^-$ ``cancel out''
until one side ``wins''. If the positive (resp.\ negative) side wins,
\ie if $ax - by + c > 0$ (resp.\ $ax - by + c \leq 0$), then it
signals all agents in $R$ to move to $\state{+0}$
(resp.\ $\state{-0}$). To achieve this, for every $0 \leq i \leq n$,
we add transition $\mathrm{cancel}_i: \state{+2^i}, \state{-2^i}
\mapsto \state{+0}, \state{-0}$ to the protocol. Since bits of the
positive and negative numbers may not be ``aligned'', we follow the
idea of Section~\ref{sec:leaderless} and add further transitions
to change representations to equivalent ones:
\begin{align*}
  \mathrm{up}_i^+:\ \state{+2^i}, \state{+2^i} &\mapsto \state{+2^{i+1}}, \state{+0}, &
  \mathrm{down}_{i+1,r}^+:\ \state{+2^{i+1}}, r &\mapsto \state{+2^i}, \state{+2^i},\\
  \mathrm{up}_i^-:\ \state{-2^i}, \state{-2^i} &\mapsto \state{-2^{i+1}}, \state{-0}, &
  \mathrm{down}_{i+1,r}^-:\ \state{-2^{i+1}}, r &\mapsto \state{-2^i}, \state{-2^i},
\end{align*}
where $0 \leq i < n$ and $r \in R$. Finally, for every $0 \leq i \leq
n$, we add transitions to signal which side wins:
\begin{align*}
  \mathrm{signal}_i^+:\ \state{+2^i}, \state{-0} &\mapsto \state{+2^i}, \state{+0}, &
  \mathrm{signal}:\ \state{-0}, \state{+0} &\mapsto \state{-0}, \state{-0}, \\
  \mathrm{signal}_i^-:\ \state{-2^i}, \state{+0} &\mapsto \state{-2^i}, \state{-0}.
\end{align*}
Note that $\state{-0}$ ``wins'' over $\state{+0}$ because the predicate is
false whenever $ax - by + c = 0$. It remains to specify the output
mapping of $\PP_\text{lin}$ which we define as expected, \ie $O(q)
\defeq 1$ if $q \in Q^+ \cup \{\state{+0}\}$, and $O(q) \defeq 0$
otherwise.

Let us briefly explain why $4n + 2$ reservoir states suffice. At any
reachable configuration $C$, transitions of the form $\mathrm{up}_i^+$
and $\mathrm{up}_i^-$ can occur until $C(\state{\pm 2^i}) \leq 1$ for
every $0 \leq i < n$. Afterwards, at most $2n$ agents remain in these
states. There can however be many agents in $S = \{\state{+2^n},
\state{-2^n}\}$. But, these two states represent numbers respectively
larger and smaller than any coefficient, hence the number of agents in
$S$ can only grow by one each time a state from $X$ is
consumed. Overall, this means that $C \trans{*} C'$ for some $C'$ such
that $C'(R) \geq 2n$.

In order to handle more variables $\{x_1, x_2, \ldots, x_k\}$, note
that all we need to do is to set $X = \{\state{x_1}, \state{x_2},
\ldots, \state{x_k}\}$ instead, and add transitions
$\mathrm{add}_{\state{x_i}, r}$ for every $1 \leq i \leq k$ and $r \in
R$.

By applying Lemma~\ref{lem:kway} on $\PP_\text{lin}$, we obtain:

\begin{restatable}{theorem}{thmLinIneq}\label{thm:lin:ineq}
  Let $a_1, a_2, \ldots, a_k, c \in \Z$ and let $n = \size{\max(|a_1|,
    |a_2|, \ldots, |a_k|, |c|, 1)}$. There exists a $2$-way population
  protocol, with at most $10kn$ states and at most $5n + 2$ leaders,
  that computes the predicate $\sum_{1 \leq i \leq k} a_i x_i + c >
  0$.
\end{restatable}

%% \subsection{Conjunction of linear inequalities}

\noindent\textbf{Conjunction of linear inequalities.} We briefly
explain how to lift the construction for linear inequalities to
systems of linear inequalities. The details of the formal construction
and proofs are a bit involved, and are thus deferred to
Appendix~\ref{app:sys:lin}. Let us fix some $A \in \Z^{m \times k}$
and $\vec{c} \in \Z^m$. We sketch a protocol $\PP_\text{sys}$ for the
predicate $A\vec{x} + \vec{c} > \vec{0}$. For every $1 \leq i \leq m$,
we construct a protocol $\PP_i$ for the predicate $\sum_{1 \leq j \leq
  k} A_{i,j} \cdot \vec{x}_j + \vec{c}_i > 0$. Protocol $\PP_i$ is
obtained as presented earlier, but with some modifications. The
largest power of two is picked as $n \defeq \size{d} + \lceil \log
2m^2 \rceil$ where
$$d \defeq \max(1, \{|A_{i,j}| : 1 \leq i \leq m, 1 \leq j \leq k\},
\{|c_i| : 1 \leq i \leq m\}).$$ The reason for this modification is
that the number of agents, in a largest power of two, should now
increase by at most $1/m$ each time an initial state is consumed, as
opposed to $1$.

We also replace each \emph{positive state} $q \in Q^+$ of $\PP_i$ by
two states $q_0$ and $q_1$, its \emph{0-copy} and \emph{1-copy}. The
reason behind this is that positive states should not necessarily have
output $1$. Indeed, one linear inequality may be satisfied while the
other ones are not. Therefore, $\state{-0}$ and each \emph{negative
  state} $q \in Q^-$ should be able to signal a $0$-consensus to the
positive states. The transitions of the form $\mathrm{up}_j^+$,
$\mathrm{down}_j^+$ and $\mathrm{cancel}_j$ are adapted accordingly.

Protocol $\PP_\text{sys}$ is obtained as follows. First, subprotocols
$\PP_1, \PP_2, \ldots, \PP_m$ are put side by side. Their initial
(resp.\ reservoir) states are merged into a single set $X$
(resp.\ $R$). For every $1 \leq j \leq k$, transitions
$\mathrm{add}_{\state{x_j}, r}$ of the $m$ subprotocols are replaced
by a single transition consuming $\state{x_j}$, and enough reservoir
states, and producing $\rep{A_{i, j}}$ in each subprotocol $\PP_i$,
where $1 \leq i \leq m$. The signal mechanisms are replaced by these
new ones:
\begin{itemize}
\item the $0$-copy of state $\state{+2^0}$ of \emph{all} subprotocols
  can meet to convert $\state{-0}$ to $\state{+0}$,

\item state $\state{+0}$ can convert any positive state to its
  $1$-copy,

\item state $\state{-0}$ or any negative state can convert
  $\state{+0}$ to $\state{-0}$, and any positive state to its $0$-copy.
\end{itemize}

A careful analysis of the formal construction of $\PP_\text{sys}$
combined with Lemma~\ref{lem:kway} yields:

\begin{restatable}{theorem}{thmSysLin}\label{thm:sys:lin}
  Let $A \in \Z^{m \times k}$, $\vec{c} \in \Z^m$ and $n =
  \size{\max(1, \{|A_{i,j}| : 1 \leq i \leq m, 1 \leq j \leq k\}},
  \linebreak \{|c_i| : 1 \leq i \leq m\})$. There exists a 2-way
  population protocol, with at most $27(\log m + n)(m + k)$ states and
  at most $14m(\log m + n)$ leaders, that computes the predicate
  $A\vec{x} + \vec{c} > \vec{0}$.
\end{restatable}

%% file: sec-conclusion.tex
We have initiated the study of the state space size of population
protocols as a function of the size of the predicate they
compute. Previous lower bounds were only for single predicates, like
the majority predicate $x \leq y$, or for a variant of the model in
which the number of states is a function of the number of agents.

%We have shown that systems of linear inequalities of size $m$ (with
%coefficients written in binary) can be evaluated by protocols with
%$O(m)$ states and $O(m)$ leaders.  Further, and surprisingly, there
%are protocols with $k$ states and 2 leaders computing $x \geq c$ for
%some $c \geq 2^{2^k}$. This result can be visualized as follows:
%imagine we add agents to a pool in an ``adiabatic process'', i.e.,
%letting agents reach consensus before adding the next one. We will
%observe a ``phase transition'' (in which all agents move from 0 to 1)
%after adding the $c$-th agent, and not earlier, even though no agent
%can store even one bit of $c$.  Finally, we have proved essentially
%matching lower bounds for 1-aware protocols: leaderless 1-aware
%protocols need at least linearly many states in the size of the
%system, and protocols with constant number of leaders need at least a
%pol-ylogarithmic number of states.

There are many open questions. We conjecture that systems of linear
inequalities can be computed by leaderless protocols with a polynomial
number of states. A second, very intriguing question is whether the
function $f(n)$ giving the minimal number of states of a two-leader
protocol computing $x \geq n$ exhibits large gaps, \ie, if there are
(families of) numbers $c$ and $c+1$ such that $f(c)$ is exponentially
larger than $f(c+1)$. A third question is whether there exist
protocols with $O(\log\log\log n)$ states for the flock-of-birds
predicates $x \geq n$. Such protocols cannot be 1-aware, but they might
exist. Their existence is linked to the long standing question of
whether the reachability problem for reversible VAS (a model
equivalent to the commutative semigroup representations
of \cite{MM82}) has the same complexity as reachability for arbitrary
VAS (see \cite{FL15} for a brief introduction).

%% file: app-k-way.tex
Let $\PP = (Q, T, I, L, O)$ be a $k$-way population protocol. We
 construct a 2-way population protocol $\PP'$ from $\PP$. For every
 transition $t : q_1, \ldots, q_k \mapsto r_1, \ldots, r_k$ where $k >
 2$, we add new \emph{disabled} states $D^t \defeq
\set{d^t_{1}, \ldots d^t_{{k-2}}}$, \emph{active} states $A^t \defeq
\set{a^t_{1}, \ldots, a^t_{{k-1}}}$ and \emph{backward} states $B^t
\defeq \set{b^t_{2}, \ldots, b^t_{k-1}}$. Consider the following
transitions, where $2 \leq \ell \leq k - 2$,
\begin{alignat*}{6}
  \mathrm{forth}_1^t \colon\ & q_1, q_2 &&\mapsto d^t_{1},
  a^t_{2}\hspace{15pt} &
  \mathrm{back}_1^t \colon\ & d^t_1, b^t_2 &&\mapsto r_1, r_2\hspace{15pt} &
  \mathrm{success}^t \colon\ & a^t_{k-1}, q^t_{k} &&\mapsto b^t_{k-1}
  r_{k} \\
  \mathrm{forth}_\ell^t \colon\ & a^t_{\ell}, q^t_{\ell+1} &&\mapsto
  d^t_{\ell}, a^t_{\ell+1}\hspace{15pt} &
  \mathrm{back}_\ell^t \colon\ & d^t_\ell, b^t_{\ell+1} &&\mapsto
  b^t_\ell, r_{\ell+1}.
\end{alignat*}

We define the \emph{inverse} of a transition $t$ as $t^{-1} \defeq
\postm{t} \mapsto \prem{t}$. We will replace every transition $t$ by
the set of transitions $T^t \defeq \mathrm{Fwd}(t) \cup
\mathrm{Fwd}^{-1}(t) \cup \set{\mathrm{success}^t} \cup
\mathrm{Bwd}(t)$ where
\begin{align*}
  \mathrm{Fwd}(t) &\defeq \set{\mathrm{forth}^t_i \colon 1 \leq i
    \leq k - 1}, &
  \mathrm{Bwd}(t) &\defeq \set{\mathrm{back}^t_i \colon 1 \leq i
    \leq k - 1}, \\
  \mathrm{Fwd}^{-1}(t) &\defeq \set{f^{-1} \colon f \in
    \mathrm{Fwd}(t)}.
\end{align*}
The transitions of $T^t$ are illustrated in
Figure~\ref{fig:kway}. Observe that a $k$-way transition $t$ can be
simulated through the following sequence of $2$-way transitions:
\begin{align*}
  \sigma_t \defeq \mathrm{forth}_1^t\ \mathrm{forth}_2^t \cdots
  \mathrm{forth}_{k-2}^t\ \mathrm{success}^t\ \mathrm{back}_{k-2}^t\ \mathrm{back}_{k-3}^t
  \cdots \mathrm{back}_1^t.
\end{align*}

Intuitively, the transitions in $\mathrm{Fwd}(t)$ temporarily
``disable'' all states of $\prem{t}$. The index $i$ of the current
active state $a_i$ keeps track of the progress that has been made in
disabling the states of $\prem{t}$. Once transition
$\mathrm{success}^t$ occurs, it is guaranteed that all states from
$\prem{t}$ have been disabled and, from this point, transition $t$ is
simulated backward through the transitions of $\mathrm{Bwd}(t)$,
transforming disabled states into $\postm{t}$.  Similarly, the index
$i$ of the backward state $b_i$ keeps track of the progress that has
been made in transforming disabled states into their respective states
of $\postm{t}$. Note that a simulation attempt may be unsuccessful,
\eg, because not all states from $\prem{t}$ are initially present in
the configuration. Unsuccessful attempts pose no problem as they can
be undone by $\mathrm{Fwd}^{-1}(t)$.

Formally, $\PP'$ is defined as $\PP' \defeq (Q', T', I, L, O')$ where
\begin{align*}
  Q' &\defeq Q \cup \bigcup_{t \in T} (D^t \cup A^t \cup B^t), \\
  T' &\defeq \bigcup_{t \in T} T^t,
\end{align*}
$O'(q) \defeq O(q)$ for every $q \in Q$, and $O(d_i^t) = O(a_i^t)
\defeq O(q_i)$ and $O(b_i^t) \defeq O(r_i)$ for every transition $t:
q_1, q_2, \ldots, q_k \mapsto r_1, r_2, \ldots, r_k$ of $T$.

\begin{figure}[h!]
  \centering
  \begin{tikzpicture}[->, node distance=1.25cm, auto, very thick, scale=0.85, transform shape]
    \node[place] (q1) {};
    \node[place, above of=q1] (q2) {};
    \node[place, above of=q2] (q3) {};
    \node[transition, right=1.75cm of q1] (t1) {};
    \node[transition, right=1.75cm of q2] (t2) {};
    \node[place, right=1.75cm of t1] (d1) {};
    \node[place, right=1.75cm of t2] (a2) {};
    \node[transition, right=1.75cm of a2] (t3) {};
    \node[place, right=1.75cm of t3] (b2) {};
    \node[place, right=3.5cm of b2] (r2) {};
    \node[place, above of=r2] (r3) {};
    \node[place, below of=r2] (r1) {};
    \node[transition, left=1.75cm of r1] (t4) {};

    \path[->]
    (q1) edge node {} (t1)
    (q2) edge node {} (t1)
    (t1) edge node {} (d1)
    (t1) edge node {} (a2)

    (d1) edge node {} (t2)
    (a2) edge node {} (t2)
    (t2) edge node {} (q1)
    (t2) edge node {} (q2)

    (a2) edge node {} (t3)
    (q3) edge[bend left=15] node {} (t3)
    (t3) edge node {} (b2)
    (t3) edge[bend left=15] node {} (r3)

    (b2) edge node {} (t4)
    (d1) edge node {} (t4)
    (t4) edge node {} (r1)
    (t4) edge node {} (r2)
    ;

    \node[] () [above= -1pt of q1] {$q_1$};
    \node[] () [above= -1pt of q2] {$q_2$};
    \node[] () [above= -1pt of q3] {$q_3$};
    \node[] () [above= -1pt of d1] {$d_1$};
    \node[] () [above= -1pt of a2] {$a_2$};
    \node[] () [above= -1pt of b2] {$b_2$};
    \node[] () [above= -1pt of r1] {$r_1$};
    \node[] () [above= -1pt of r2] {$r_2$};
    \node[] () [above= -1pt of r3] {$r_3$};

    \node[] () [above= 0pt of t1] {$\mathrm{forth}_1$};
    \node[] () [above= 0pt of t2] {$\mathrm{forth}_1^{-1}$};
    \node[] () [above= 0pt of t3] {$\mathrm{success}$};
    \node[] () [above= 0pt of t4] {$\mathrm{back}_1$};
    ;
  \end{tikzpicture}
  \caption{Gadget of 2-way transitions simulating the 3-way transition
    $q_1, q_2, q_3 \mapsto r_1, r_2, r_3$. Circles and squares depict
    respectively states and transitions.}\label{fig:kway}
\end{figure}
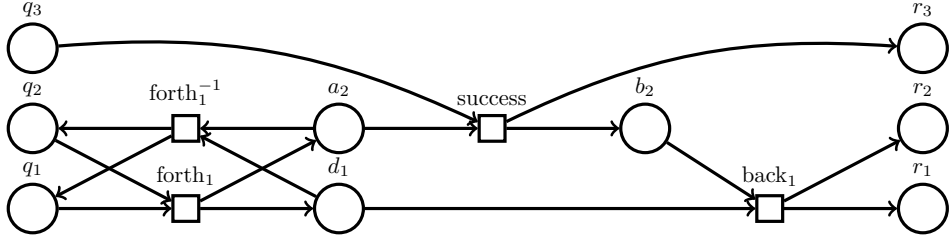

In the remainder of this appendix, we prove the following:

\lemKway*

The bound stated in Lemma~\ref{lem:kway} follows directly from the
construction. Therefore, we must only prove that $\PP'$ computes the
same predicate as $\PP$. To facilitate the proof of
Lemma~\ref{lem:kway}, we introduce a more fine-grained notion of
``simulation'' than mere equality of predicates.

Let $\PP_1=(Q_1, T_1, L_1, I_1, O_1)$ and $\PP_2=(Q_2, T_2, L_2, I_2,
O_2)$ be two well-specified population protocols. We say $\PP_2$
\emph{simulates} $\PP_1$ if the following holds:
\begin{enumerate}
\item $Q_1 \subseteq Q_2$,
\item $I_1 = I_2$ and $L_1 = L_2$,
\item $O_1(q) = O_2(q)$ for every $q \in Q_1$, \label{prop_3}
\item $C \trans{*}_1 C' \Leftrightarrow C \trans{*}_2
  C'$ for every $C, C' \in \pop{Q_1}$, \label{prop_4}
\item $\forall C \in \pop{Q_1}, C' \in \pop{Q_2} \colon C \trans{*}_2
  C' \Rightarrow \exists C'' \in \pop{Q_1} \colon C' \trans{*}_2 C''
  \land C \trans{*}_1 C''$.
 \label{prop_5}
\end{enumerate}

Before proving Lemma~\ref{lem:kway}, let us first show that the above
notion of simulation indeed implies equality of predicates:
\begin{proposition}
  Let $\PP_1$ and $\PP_2$ be two well-specified protocols.  If
  $\PP_2$ simulates $\PP_1$, then $\PP_1$ and $\PP_2$ compute the
  same predicate.
\end{proposition}

\begin{proof}
  Let $\pi_1$ and $\pi_2$ be fair executions of $\PP_1$ and $\PP_2$,
  respectively, both starting from some initial configuration $C_0 \in
  \pop{I_1} = \pop{I_2}$. Since $\PP_1$ and $\PP_2$ are
  well-specified, there exist $b_1, b_2 \in \set{0, 1}$ such that
  $O_1(\pi_1) = b_1$ and $O_2(\pi_2) = b_2$. It remains to show that
  $b_1 = b_2$.  By fairness and Property~\ref{prop_5}, there exists
  some configuration $C \in \pop{Q_1}$ that occurs infinitely often in
  $\pi_2$. By Property~\ref{prop_3}, Property~\ref{prop_4} and
  well-specification of $\PP_1$, configuration $C$ must be stable in
  $\PP_1$. Moreover, $C$ must be reachable from $C_0$ in $\PP_1$ by
  Property~\ref{prop_4}. Thus, due to well-specification of $\PP_1$,
  we have $O_1(\pi_1) = O_1(C)$. By Property~\ref{prop_3}, we also
  know that $O_1(C) = O_2(C)$ must hold. Consequently $b_1 = O(\pi_1)
  = O_1(C) = O_2(C) = O_2(\pi_2) = b_2$.
\end{proof}

It remains to prove that if $\PP$ is well-specified, then so is
$\PP'$, and that $\PP'$ simulates $\PP$. We first show the latter.
Properties 1--3 are cleary satisfied. To show the remaining
properties~\ref{prop_4} and~\ref{prop_5}, fix some $n \in \N$, $C_0,
C_1, \ldots, C_n \in \pop{Q'}$ and $t_1, t_2, \ldots, t_n \in T'$ such
that $C_0 \in \pop{Q}$ and $$C_0 \trans{t_1} C_1 \trans{t_2} \ldots
\trans{t_{n}} C_n.$$ We define $H$ as the set of \emph{helper states}
of $Q'$, \ie, $$H \defeq Q' \setminus Q.$$ Whenever an agent changes
its state from $Q$ to $H$, the agent can be thought of as
participating in a simulation attempt of some $k$-way transition that
was started at some point in time $x \in [n]$. In order to make this
association explicit, we annotate the helper states with timestamps
from $[n]$, \ie, we augment $H$ to $\hat{H} \defeq H \times [n]$. We
also augment every transition $t \colon \multiset{q_1, q_2 } \mapsto
\multiset{r_1, r_2}$ of $T'$ with timestamps $x \in [n]$, \ie, $t^x
\colon \multiset{q^x_1, q^x_2 } \mapsto \multiset{r^x_1, r^x_2}$ where
for every $q \in Q'$, $q^x$ is defined as:
    $$q^x \defeq \begin{cases}
      q      & \text{ if } q \in Q, \\
      (q, x) & \text{ otherwise.}
    \end{cases}$$

We now inductively define an execution $\hat{C}_0 \trans{\hat{t}_1}
\ldots \trans{\hat{t}_n} \hat{C}_n$ augmented by timestamps. Let
$\hat{C}_0 \defeq C_0$. For every $i \in [n-1]$, let
\begin{align*}
  a(i) & \defeq
  \begin{cases}
    i & \text{ if } \pre{t_i} \subseteq Q, \\
    \text{smallest } j \text{ s.t. } \prem{t_i^j} \leq \hat{C}_{i-1} &
    \text{ otherwise},
  \end{cases} \\
  \hat{t}_i & \defeq t_i^{a(i)}, \\
  \hat{C}_i & \defeq \hat{C}_{i-1} - \prem{\hat{t}_i} + \postm{\hat{t}_i}.
\end{align*}
Intuitively, $a(i)$ denotes the timestamp of the beginning of the
simulation attempt which transition $t_i$ belongs to. If $t_i$ could
belong to several simulation attempts, then we pick the earliest one.

For every $x \in [n]$, let $\hat{C}_i(x) \in \pop{H}$ denote the
configuration resulting from extracting all helper states labelled by
$x$ from $\hat{C}$, \ie, $\left(\hat{C}_i(x)\right)\left(h\right)
\defeq \hat{C}_i((h, x))$ for every $h \in H$.
\begin{proposition}\label{prop_welldefined}
  For every $i \in [0, n]$ the following holds:
  \begin{itemize}
  \item $\hat{C}_i$ is well-defined.
  \item If $i > 0$, then $a(i)$ and $\hat{t}_i$ are well-defined.
  \item For every $x \in [n]$, there exists a transition $t \in T
    \colon q_1, \ldots q_n \mapsto r_1, \ldots, r_n$ and some $\ell < n$
    such that if $\hat{C}_i(x) \neq \mathbf{0}$, then $\hat{C}_i(x) =
    \multiset{d_1^t, d_2^t, \ldots, d_{\ell-1}^t, a_\ell^t}$ or $\hat{C}_i(x) =
    \multiset{d_1^t, d_2^t, \ldots, d_{\ell-1}^t, b_\ell^t}$.
  \item $\hat{C}_i \cap \pop{Q} = C_i \cap \pop{Q}$.
  \end{itemize}
\end{proposition}

\begin{proof}
  The proof is by induction on $i$. Configuration $\hat{C}_0$ is
  clearly well-defined. Moreover, $\hat{C}_0(x) = \vec{0}$ for every
  $x \in [n]$ and $\hat{C}_0 = C_0$, and hence the third and fourth
  points hold trivially.

  Let $i > 0$ and assume the claim holds for all values smaller than
  $i$. Let $t \colon q_1, \ldots, q_k \mapsto r_1, \ldots, r_k \in T$
  be the transition that is simulated by $t_i$, \ie such that $t_i \in
  T^t$. We make the following case distinction:
  \begin{itemize}
  \item \emph{Case 1: $t_i = \text{forth}_1^t$.} By definition of
    $t_i$, we have $\pre(t_i) \subseteq Q$. Thus, $a(i)$ and
    $\hat{t}_i$ are obviously well-defined. Note that
    $\prem{\hat{t}_i} = \multiset{q_1, q_2} = \prem{t_i} \leq
    C_{i-1}$. By induction hypothesis, $\hat{C}_{i-1} \cap \pop{Q} = C_{i-1}
    \cap \pop{Q}$. In particular, this implies that $\pre{(\hat{t_i})}
    \leq \hat{C}_{i-1}$ which in turn implies that $\hat{C}_i$ is
    well-defined. The third point holds since $\hat{C}_i(i) =
    \multiset{d_1^t, a_1^t}$. The fourth point holds since
    \begin{align*}
      \hat{C}_i \cap \pop{Q}
      &= ((\hat{C}_{i-1} \mminus \prem{\hat{t}_i}) \mplus
      \postm{\hat{t}_i}) \cap \pop{Q} \\
      &= ((\hat{C}_{i-1} \cap Q) \mminus (\prem{\hat{t}_i} \cap Q))
      \mplus (\postm{\hat{t}_i} \cap \pop{Q}) \\
      &= ((C_{i-1} \cap Q) \mminus (\prem{t_i} \cap Q)) \mplus
      (\postm{t_i} \cap \pop{Q}) \\
      &= ((C_{i-1} \mminus \prem{t_i}) \mplus \postm{t_i}) \cap \pop{Q} \\
      &= C_i \cap \pop{Q}.
    \end{align*} 
  
  \item \emph{Case 2: $t_i = \text{forth}_\ell^t$ for some $1 < \ell <
    k$.}  Recall that $t_i \colon a^t_\ell, q_{\ell+1}^t \mapsto
    d^t_\ell, a^t_{\ell+1}$. Since $t_i$ is enabled at $C_{i-1}$, we
    have that $C_{i-1}(a^t_\ell) > 0$. Thus, there exists some $x \in
    [n]$ such that $\hat{C}_{i-1}((a^t_\ell, x)) > 0$. Pick $x$ as the
    smallest such number. By induction hypothesis, $\hat{C}_{i-1}(x) =
    \multiset{d_1^t, \ldots, d_{k-1}^t, a_k^t}$ for some $k <
    n$. Since $\hat{C}_{i-1}((a^t_\ell, x)) > 0$, we must have $k =
    \ell$. Thus, $\prem{t_i^x} \leq \hat{C}_{i-1}$. Now, observe that
    $a(i) = x$, and hence that both $a(i)$ and $\hat{C}_i$ are
    well-defined. The third point holds since $\hat{C}_{i}(x) =
    \multiset{d_1^t, d_2^t, \ldots, d_{\ell}, a_{\ell+1}^t}$. The
    proof of the fourth point is the same as in case 1.\medskip
  
  \item \emph{Case 3: $t_i = \text{success}^t$ or $t_i =
    \text{back}_\ell^t$.} The reasoning is analogous to the last
    case.\qedhere
  \end{itemize}	
\end{proof}

For every $i \in [n]$, we say that $a(i)$ is \emph{successful} if
there exist $j \in [n]$ and $t \in T$ such that $a(i) = a(j)$ and $t_j
= \text{success}^t$. It can be shown that index $j$ must be unique. We
denote this index $j$ by $s(i)$.

We now state three useful propositions whose proofs are left to the
reader. Let $\text{Fwd}^{-1} \defeq \bigcup_{t \in T}
\text{Fwd}^{-1}(t)$.
\begin{proposition}
  For every $i \in [n]$, the following holds:
  \begin{itemize}
  \item\label{prop_undo_succ} If $t_{i} \not \in \text{Fwd}^{-1}$ and
    $Q \cap \post{(\hat{t}_{i})} \neq \emptyset$, then $s(i) \leq i$.

  \item\label{prop_succ} If $Q \cap \pre{(\hat{t}_{i})} \neq
    \emptyset$, then $s(i) \geq i$.
  \end{itemize}
\end{proposition}

\begin{proposition}
  Let $C_0, C \in \pop{Q}$ be such that $C_0 \trans{*}_{\PP'} C$. The
  following holds:
  \begin{itemize}
  \item \label{prop_always_no_undo}
    $C$ is reachable from $C_0$ in $\PP'$ without using transitions from 
    $\text{Fwd}^{-1}$.  
  \item \label{prop_always_succ} There exist $C_1, C_2, \ldots, C_n
    \in \pop{Q'}$ and $t_1, t_2, \ldots, t_n \in T'$ such that $C_0
    \trans{t_1} C_1 \trans{t_2} \cdots \trans{t_n} C_n = C$ such
    that, for every $i \in [n]$, $a(i)$ is successful in the augmented
    execution $\hat{C}_0 \trans{\hat{t}_1} \hat{C}_1 \trans{\hat{t}_2}
    \ldots \trans{\hat{t}_n} \hat{C}_n$.
  \end{itemize}
\end{proposition}

\begin{proposition}\label{prop_simulation}
  Let $t \in T$ and $\sigma \in {\left(T^t \setminus
    \set{\text{forth}_1^t, \left(\text{forth}_1^t \right)^{-1}}\right)}^*$ and $C, C' \in \pop{Q}$. If $C
  \trans{\text{forth}_1^t \cdot \sigma} C'$, then $C \trans{t} C'$.
\end{proposition}

The following lemma shows that the execution order of two transitions
belonging to different simulation attempts can be swapped under
certain conditions:

\begin{lemma}\label{lemma_commute_sim}
  Let $i \in [n-1]$ be such that $a(i)$ and $a(i+1)$ are both successful
  simulation attempts satisfying $s(i+1) < s(i)$. If $t_i \not \in
  \text{Fwd}^{-1}$, then $\hat{C}_{i-1} \trans{\hat{t}_{i+1}
    \hat{t}_{i}} \hat{C}_{i+1}$.
\end{lemma}

\begin{proof}
  For the sake of contradiction assume $C_{i-1} \trans{\hat{t}_{i+1}
    \hat{t}_{i}} C_{i+1}$ does not hold. This entails that
  \begin{equation}\label{eq_conflict}
    \post{\hat{t}_{i}} \cap \pre{\hat{t}_{i+1}} \neq \emptyset
  \end{equation}
  Moreover, since $s(i+1) < s(i)$, we have that $a(i+1) \neq
  a(i)$. Thus
  \begin{equation}\label{eq_no_conflict}
    \post{\hat{t}_{i}} \cap \pre{\hat{t}_{i+1}} \cap \hat{H} = \emptyset
  \end{equation}
  Inequality~\eqref{eq_conflict} and Equality~\eqref{eq_no_conflict}
  combined then yield
  \begin{equation}\label{ineq_final}
    Q \cap \post{\hat{t}_{i}} \cap \pre{\hat{t}_{i+1}} \neq \emptyset
  \end{equation}
  
  Since $t_i \not \in \text{Fwd}^{-1}$ by assumption, we obtain from
  Proposition~\ref{prop_undo_succ} and Inequality~\eqref{ineq_final}
  that $s(i) \leq i$. Moreover, Inequality~\eqref{ineq_final} and
  Proposition~\ref{prop_succ} imply that $s(i+1) \geq i+1$. Thus $s(i)
  < s(i+1)$, which contradicts our initial assumption that $s(i+1) <
  s(i)$.
\end{proof}

\begin{corollary}
  Property~\ref{prop_4} holds.
\end{corollary}

\begin{proof}
  Fix some $C, C' \in \pop{Q}$ and let $\PP_1 = \PP$ and $\PP_2 =
  \PP'$.

  $\Rightarrow$) Assume $C \trans{*}_1 C'$. We have to show that $C
  \trans{*}_2 C'$ holds. We saw earlier how a single $k$-way
  transition of $\PP_1$ can be simulated via a sequence of $2$-way
  transitions of $\PP_2$. Thus, ${\trans{*}_1} \subseteq
  {\trans{*}_2}$ and we are done.

  $\Leftarrow$) Assume $C \trans{t_1} C_1 \trans{t_2} \ldots
  \trans{t_{n}} C'$ for some $t_1, t_2, \ldots, t_n \in T'$. Consider
  the augmented run $\hat{C} \trans{\hat{t}_1} \hat{C}_1
  \trans{\hat{t}_2} \ldots \trans{\hat{t}_{n}} \hat{C}'$. By
  Proposition~\ref{prop_always_succ}, we
  may assume that $a(i)$ is successful and $t_i \not\in
  \text{Fwd}^{-1}$ for every $i \in [n]$.

  Let $A \defeq \set{a(i) : i \in [n]}$ be the set of successful
  simulation attempts and let $m \defeq |A|$. By repeatedly applying
  Lemma~\ref{lemma_commute_sim}, we can reorder the augmented
  execution such that $\hat{C} \trans{T_1} C_1' \trans{T_2} C_2'
  \trans{T_3} \ldots \trans{T_m} \hat{C}'$ for some $C_i'$, where each
  $T_i$ is a sequence of transitions that belong to exactly one of the
  successful simulation attempts, \ie
	$$ T_i \in \set{{t'}^x \colon t' \in T^t}^* $$ for some $x \in
  A$ and $t \in T$. Observe that $C_i' \in \pop{Q}$ for every $i \in
  [m]$, and moreover $C_i \trans{}_1 C_{i+1}$ for every $i \in [m-1]$:
  By Proposition~\ref{prop_simulation}, each sequence $T_i$
  corresponds to the successful simulation of some $k$-way transition
  that must be enabled at $C_{i-1}$. Thus $C \trans{*}_1 C'$, which
  completes the proof for Property~\ref{prop_4}.
\end{proof}

In order to show Property~\ref{prop_5}, we only need to show that
every execution of $\PP'$ can be extended to an execution that ends up
in a configuration without helper states. Validity of
Property~\ref{prop_5} then follows from Property~\ref{prop_4}. The
following lemma proves a slightly stronger result.

\begin{lemma}\label{lemma_exec_extension}
  Let $C_0, C_1, \ldots, C_n \in \pop{Q'}$ and let $t_1, t_2, \ldots,
  t_n \in T'$ be such that $C_0 \in \pop{Q}$ and $C_0 \trans{t_1} C_1
  \trans{t_2} \cdots \trans{t_{n}} C_n$. There exists some $C' \in
  \pop{Q}$ such that $C_n \trans{*}_{\PP'} C'$ and $Q \cap \supp{C_n}
  \subseteq \supp{C'}$.
\end{lemma}

\begin{proof}
  Consider the augmented run $\hat{C}_0 \trans{\hat{t}_1} \hat{C}_1
  \trans{\hat{t}_2} \cdots \trans{\hat{t}_{n}} \hat{C}_n$. If
  $\hat{C}_n \in \pop{Q}$, then we are done. Otherwise every helper
  state in $\supp{\hat{C}_n}$ is labelled by some simulation
  attempt. Let $x_1 \leq x_2 \leq \ldots \leq x_m$ be these simulation
  attempts, \ie let
  $$\set{x_1, x_2, \ldots, x_m} = \set{x \in [n] : (H \times \set{a})
    \cap \supp{\hat{C}_n} \neq \emptyset}.$$ By
  Proposition~\ref{prop_welldefined}, one of two cases must hold:
  either (1) $\hat{C}_n(x_i) = \multiset{d_1^t, \ldots, d_{\ell-1}^t,
    a_\ell^t}$ or (2) $\hat{C}_n(x_i) = \multiset{d_1^t, \ldots,
    d_{\ell-1}^t, b_\ell^t}$ for some $\ell < n$ and $t \in T$. For
  each attempt $x_i$, we construct a sequence of transitions $T(x_i)$
  as follows:\medskip

  \noindent\emph{Case 1.} We construct $T(x_i) \defeq
  (\left(\text{forth}^{t}_{\ell})^{-1}\right)^{x_i} \cdots
  \left((\text{forth}^t_1)^{-1} \right)^{x_i}$. In this case, the
  sequence $T(x_i)$ ``undoes'' the unsuccessful simulation attempt
  $x_i$.\medskip

  \noindent\emph{Case 2.} We construct $T(x_i) \defeq
  \left(\text{back}^t_\ell\right)^{x_i} \cdots
  \left(\text{back}^t_1\right)^{x_i}$. In this case, $T(x_i)$
  ``completes'' the successful simulation attempt $x_i$.\medskip

  Observe that $\hat{C}_n \trans{T(x_i)} C'$ implies that $\supp{C'}
  \cap (H \times \set{a_i}) = \emptyset$. Also, note that $T(x_k)$ and
  $T(x_i)$ can occur independently for $k \neq i$, as the presets of
  $t_i$ and $t_k$ contained in $T(x_i)$ and $T(x_k)$ are disjoint, and
  their presets solely contain helper states which are labelled by
  different simulation attempts:
  $$\pre{t_i}\cap \pre{t_k} = \pre{t_i} \cap Q = \pre{t_k} \cap Q =
  \emptyset.$$ Thus, there exists some $C' \in \pop{Q}$ satisfying
  $\hat{C}_n \trans{T(x_1) T(x_2) \cdots T(x_n)} C'$ and
  $\supp{\hat{C}_n} \cap Q \subseteq \supp{C'}$.

  Let
  $\pi \colon \hat{C}_0 \trans{\hat{t}_1} \hat{C_1} \trans{\hat{t_2}} \cdots \trans{\hat{t}_n} \hat{C}_n \trans{T(x_1)
  T(x_2) \cdots T(x_n)} C'$. Execution $\pi$ can be ``projected'' by
  removing the timestamps of its configurations and transitions. By
  definition of augmented executions, this projection yields an
  execution from $C_0$ to $C'$ in $\PP'$, which proves the claim.
\end{proof}

\begin{corollary}
  Property \ref{prop_5} holds.
\end{corollary}

It remains to show that $\PP'$ is well-specified if $\PP$ is
well-specified.

\begin{proposition}
  If $\PP$ is well-specified, then $\PP'$ is also well-specified.
\end{proposition}

\begin{proof}
  Let $\PP$ be a well-specified $k$-way protocol. For contradiction
  assume the simulating protocol $\PP'$ was not well-specified. This
  means either of two things must hold:
  \begin{itemize}
  \item There exist two fair executions $\pi_1$ and $\pi_2$ starting
    in the same initial configuration and such that $O(\pi_1) \neq
    O(\pi_2)$.
  \item There exists a fair execution $\pi$ starting in an initial
    configuration such that $O(\pi) = \bot$.
  \end{itemize}
  We only show that the validity of the second claim leads to a
  contradiction. The proof can easily be adapted to arrive at a
  contradiction for the first claim. Assume there exists a fair
  execution $\pi = C_0 C_1 C_2 \cdots$ of $\PP'$ starting in some
  initial configuration $C_0$ and such that $O(\pi) = \bot$. Due to
  well-specification of $\PP$, Proposition~\ref{prop_5}
  and~\ref{prop_4} and fairness, we know this execution will reach a
  configuration $C_i$ that is stable in $\PP$. Let $i \in \N$ be the
  smallest such index. Moreover, let $j$ be the smallest index larger
  than $i$ such that $O(C_j) \neq O(C_i)$. Since $\pi$ does not
  stabilize, such an index $j$ must exist. Observe that whenever an
  agent changes from a non-helper state to a helper-state, or from a
  helper-state to a helper-state, outputs do not change. Thus, it must
  hold that $C_{j-1} \trans{\text{success}^t} C_j$ for some $t \in T$,
  for only in this case an agent changes from a helper state to some
  non-helper state $q$ of output $O(q) \neq O(C_i)$. By
  Lemma~\ref{lemma_exec_extension}, there exists some configuration
  $C' \in \pop{Q^\PP}$ such that $C_i \trans{*}_{\PP'} C'$ and $q \in
  \supp{C'}$. From this and by Property~\ref{prop_4}, we have $C_i
  \trans{*}_{\PP} C'$. But $O(C') \neq O(C_i)$, which contradicts our
  assumption that $C_i$ is stable in $\PP$.
\end{proof}

%% file: app-leaderless.tex
\thmExistBound*

\begin{proof}
  We first show that for every finite family $\{\PP_0, \PP_1, \ldots,
  \PP_n\}$ of 2-way population protocols computing the predicates $\{x
  \geq 0, x \geq 1, \ldots, x \geq n\}$ there exists $0 \leq j \leq n$
  such that $\PP_j$ has at least $(\log n)^{1/4}$ states. For this, we
  prove an equivalent statement: 2-way protocols with at most $m$
  states can compute at most $2^{m^4}$ unary predicates.

  Let $d(m)$ be the number of unary predicates computed by 2-way
  population protocols with at most $m$ states. Every protocol with
  less than $m$ states can be extended to a protocol with $m$ states
  computing the same predicate, and so in order to bound $d(m)$ it
  suffices to consider protocols with exactly $m$ states.  Further,
  for the same reason, we only consider protocols containing all
  possible silent transitions, \ie, all transitions of the form $x,
  y \mapsto x, y$. Such a protocol is completely determined by its set
  of non-silent transitions, its initial state, and its output
  mapping. Since the number of sets of non-silent transitions is
  bounded by $2^{m^4 - m^2}$, the number of initial states by $m$, and
  the number of output mappings by $2^m$, there are at most $2^{m^4 -
  m^2} \cdot m \cdot 2^m$ such protocols. Altogether we obtain: 
  \begin{align*} d(m)
  &\leq 2^{m^4 - m^2} \cdot m \cdot 2^m  =
  2^{m^4} \cdot \frac{2^m \cdot m}{2^{m^2}}  \leq
  2^{m^4}.  \end{align*}
  
  Now we prove the theorem. Let $\{\PP_0, \PP_1, \ldots\}$ be an
  infinite family of 2-way protocols such that $\PP_i$ computes $x
  \geq i$ for every $i \in \N$. By the above result, for every $n \geq
  0$ there is $j_n \leq n$ such that $\PP_{j_n}$ has at least $(\log
  n)^{1/4} \geq (\log j_n)^{1/4}$ states. It remains to prove that the
  set $\{ j_0, j_1, \ldots \}$ is infinite. Let $m_i$ be the number of
  states of $\PP_{j_i}$. Since $\lim_{i\rightarrow \infty} m_i =
  \infty$, we can extract from the sequence $m_0, m_1, \ldots$ a
  strictly increasing subsequence $m_{n_1} < m_{n_2} < \cdots$. Thus,
  the indices $j_{n_1}, j_{n_2}, \ldots$ are all distinct, and we are
  done.
\end{proof}

%% file: app-loglog-upper.tex
\lemSemigroupSim*

\begin{proof}
  For~(a), the only reason why
  $\mathrm{pad}(p_1) \cdots \mathrm{pad}(p_k)$ could not occur from
  $C_{\alpha, m}$ is that this configuration may not have enough
  agents in state $x$. By~(1), the left hand side of every transition
  $\mathrm{pad}(p_i)$ removes at most $4$ agents from state $x$, and
  so $\mathrm{pad}(p_1) \cdots \mathrm{pad}(p_k)$ can occur for any
  $m \geq 4k$. Item~(b) follows immediately from the definitions.
\end{proof}

\thmLoglogUpperbound*

\begin{proof}
  We first show that $\PP_n$ is well-specified. Let $C_0$ be an
  initial configuration. We make a case distinction on whether $f_n$
  is coverable from $C_0$ or not.\medskip

  \noindent\emph{Case 1: $f_n$ is coverable.} Let $\pi = C_0 C_1
  \cdots$ be a fair execution. We claim that $C_i(f_n) > 0$ for
  infinitely many indices $i$. The claim proves the case since
  fairness and transitions of $T_n^2$ ensure that all agents
  eventually remain in $f_n$, and hence that $O(\pi) = 1$.

  For the sake of contradiction, assume the claim does not hold. Let
  $i \in \N$ be the minimal index such that $C_{i}(f_n) = C_{i+1}(f_n)
  = \cdots = 0$. If $i = 0$, then $\pi$ only consists of transitions
  of $T_n^1$. By assumption, $C_0 \trans{*} C$ for some configuration
  $C$ such that $C(f_n) > 0$. Note that $C$ does not occur in $\pi$.
  We make use of the reversibility property of $\semig{n}$. Since
  $\alpha \trans{*} \beta$ if and only if $\beta \trans{*} \alpha$ in
  $\semig{n}$, by Lemma~\ref{lem:semigroup:sim} we have $C_j \trans{*}
  C_0 \trans{*} C$ for every $j \in \N$, which contradicts $\pi$ being
  fair. Therefore, we must have $i > 0$. Let $\sigma_j$ be the
  sequence from $C_{i-1}$ to $C_j$ in $\pi$, for every $j \geq
  i$. Note that $C_{i-1}$ only occurs finitely often in
  $\pi$. Moreover, each $\sigma_j$ only contains transitions from
  $T_n^1$. Therefore, using reversibility again, we obtain $C_j
  \trans{*} C_{i-1}$ for every $j \geq i$. We derive a contradiction
  since, by fairness, $C_{i-1}$ should occur infinitely often in
  $\pi$.\medskip

  \noindent\emph{Case 2: $f_n$ is not coverable.} Let $\pi = C_0 C_1
  \cdots$ be a fair execution. Suppose $O(\pi) \neq 0$. As $f_n$ is
  the only state with output 1, there exists $i \in \N$ such that
  $C_i(f_n) > 0$. Since $C_i$ is reachable from $C_0$, state $f_n$ is
  coverable from $C_0$. This is a contradiction and hence $O(\pi) =
  0$.\medskip

  It remains to prove that $\PP_n$ computes $x \geq c_n$ for some
  number $c_n \geq 2^{2^n}$. By~(3) and Lemma~\ref{lem:semigroup:sim},
  state $f_n$ is coverable from some initial configuration $C_0$. By
  the above case 1, $O(C_0) = 1$. Let $C_0$ be the smallest such
  configuration. By~(3) and Lemma~\ref{lem:semigroup:sim}, we have
  $|C_0| \geq 2^{2^n}$. Moreover, state $f_n$ is coverable from every
  configuration larger that $C_0$. Thus, by the above case 1, we have
  $O(C_0') = 1$ for every initial configuration $C_0'$ such that
  $|C_0'| \geq |C_0|$. Therefore, the protocol computes the predicate
  $x \geq |C_0|$ where $|C_0| \geq 2^{2^n}$.
\end{proof}

\corLoglogUpperbound*

\begin{proof}
  Let $n \in \N$ and let $\PP_n' = (Q_n', {T_n^1}' \cup {T_n^2}', I_n',
  L_n', O_n')$ be the protocol of
  Theorem~\ref{thm:loglog:upperbound}. By applying
  Lemma~\ref{lem:kway} to $\PP_n'$ we obtain a 2-way protocol $\PP_n =
  (Q_n, T_n, I_n, L_n, O_n)$ such that
  \begin{itemize}
  \item $|Q_n| = |Q_n'| + 3 \cdot 5 \cdot |{T_n^1}'| \leq (14n + 11) +
    (300n + 120) = 314n + 131$,
  \item $\PP_n$ computes the same predicate as $\PP_n'$, \ie $x \geq
    c_n$ for some $c_n \geq 2^{2^n}$. \qedhere
  \end{itemize}
\end{proof}

%% file: app-one-aware-monotonic.tex
In this section, we relate $1$-aware protocols to monotonic predicates.

\begin{definition}
  Let $n \in \N$ and let $\varphi \subseteq \N^n$ be an $n$-ary
  predicate. We say $\varphi$ is \emph{monotonic} if and only if
  $(\mathbf{y} \geq \mathbf{x} \land \varphi(\mathbf{x})) \implies
  \varphi(\mathbf{y})$ for every $\mathbf{x}, \mathbf{y} \in \N^n$.
\end{definition}

\begin{proposition}\label{prop_montonicity_disjunction}
  For every monotonic predicate $\varphi \subseteq \N^n$ of arity $n
  \in \N$ there exists a finite family of thresholds
  $\set{\mathbf{c}_1, \ldots, \mathbf{c}_m} \subseteq \N^n$ such that
  $$\varphi(\mathbf{x}) \Longleftrightarrow \bigvee_{1 \leq i \leq m} \mathbf{x} \geq \mathbf{c}_i.$$
\end{proposition}

\begin{proof}
  By the very definition of monotonicity, the set $\set{\mathbf{x} :
    \varphi(\mathbf{x})}$ is upwards-closed w.r.t. $\leq$ and thus has
  a finite number $m$ of minimal elements by Dickson's lemma. Picking
  these minimal elements $\mathbf{c}_1, \ldots, \mathbf{c}_m$ as the
  finite family of thresholds then yields the claim to be shown.
\end{proof}

\begin{lemma}
  Let $n \in \N$ and let $\varphi$ be some $n$-ary predicate
  computable by a population protocol. Predicate $\varphi$ is
  computable by a $1$-aware protocol if and only if $\varphi$ is
  monotonic.
\end{lemma}

\begin{proof}
  We first show that if $\PP$ is $1$-aware, then the predicate
  $\varphi$ computed by $\PP$ is monotonic. Let $C_0, C_0'$ be initial
  configurations such that $\varphi\left(C_0\right)$ holds and $C_0
  \leq C_0'$. We must show that $\varphi(C_0')$ holds. Let $Q_1
  \subseteq Q$ be the subset of states that makes $\PP$
  $1$-aware. Since $\varphi(C_0)$ holds, there exists $q \in Q_1$ and
  a configuration $C$ such that $C_0 \trans{*} C$ and $q \in
  \supp{C}$. Since $C_0' \geq C_0$, we have $C_0' \trans{*} C'$ for
  some $C' \geq C$. This implies that $q \in \supp{C'}$. By
  $1$-awareness of $\PP$, we conclude that $\varphi(C_0')$ holds.

  For the converse direction, assume $\varphi$ is a monotonic
  predicate computable by a population protocol. By
  Proposition~\ref{prop_montonicity_disjunction}, we may assume
  $\varphi$ is a finite disjunction of predicates of the form
  $\mathbf{x} \geq c_i$ for some thresholds $c_i$. As
  threshold-predicates can be computed by $1$-aware protocols and
  $1$-aware protocols are closed under disjunction, $\varphi$ is
  computable by a $1$-aware protocol, and we are done.
\end{proof}

%% file: app-one-aware.tex
\lemSimpleOneAware*

\begin{proof}
  Let $n \in \N$ and  let $C_0 \defeq \multiset{n \cdot x}+L$.

  $\Rightarrow$) Let $\pi = C_0 C_1 \cdots$ be a fair execution. Since
  $\PP$ computes $\varphi$, we have $O(\pi) = \varphi(n) = 1$. By
  condition~(2) of the definition of 1-awareness, $C_j(Q_1) > 0$ for
  some $j \in \N$. We are done since $C_0 \trans{*} C_j$.

  $\Leftarrow$) We have $C_0 \trans{} C_1 \trans{} \cdots \trans{} C_n
  = C$ for some configurations $C_1, C_2, \ldots, C_n$. Let $\pi = C_0
  C_1 \cdots C_n \cdots$ be any fair execution extending this finite
  sequence. By condition~(1) of the definition of 1-awareness,
  $O(\pi) \neq 0$, and hence $O(\pi) = 1$.
\end{proof}

\propShortSaturation*

\begin{proof}
  Let $c \defeq k+1$. We prove a stronger claim: $C_1'$, $C_2'$, and
  $\pi'$ can be chosen so that they satisfy~($a$), ($b$), and a
  stronger property: ($d$) there is a sequence $t_1, t_2, \ldots, t_n$
  of transitions of $\supp{\pi}$ such that $\pi' = t_1^{c^{n-1}}
  t_2^{c^{n-2}} \cdots t_n$ and $n \leq |\post{\{t_1, \ldots, t_n\}}|$.

  We proceed by induction on $|\pi|$. If $|\pi| = 0$, then $\supp{\pi}
  = \emptyset$ and $C_1 = C_2$. Thus, the claim is satisfied by
  $\pi' \defeq \epsilon$ and the configurations $C_1'$ and $C_2'$ such
  that for every $q \in Q$,
  $$
  C_1'(q) \defeq C_2'(q) \defeq
  \begin{cases}
    1 & \text{if } q \in \supp{C_1}, \\
    0 & \text{otherwise}.
  \end{cases}
  $$ Assume that $|\pi| > 0$ and that the claim holds for sequences of
  length less than $|\pi|$. There exist $\sigma \in T^*$, $t \in T$
  and a configuration $D$ such that $\pi = \sigma t$ and $C_1
  \trans{\sigma} D \trans{t} C_2$. By induction hypothesis, there
  exists an execution $C_1'' \trans{\pi''} D''$ such that ($a'$)
  $\supp{C_1''} = \supp{C_1}$, ($b'$) $\supp{D''} = \supp{C_1''} \cup
  \post{\supp{\sigma}}$, and ($d'$) $\pi'' = t_1^{c^{m-1}}
  t_2^{c^{m-2}} \cdots t_m$ for a sequence $t_1, t_2, \ldots, t_m$ of
  transitions of $\supp{\sigma}$ satisfying $m \leq |\post{\{t_1,
    \ldots, t_m\}}|$.
 
  If $\post{t} \subseteq \post{\{t_1, \ldots, t_m\}}$, then we can
  take $C_1' \defeq C_1''$, $C_2' \defeq D''$, and
  $\pi' \defeq\pi''$. So assume
  $\post{t} \not\subseteq \post{\{t_1, \ldots, t_m\}}$. Since
  $C_1 \trans{\sigma} D$, we have
  $\supp{D} \subseteq \supp{C_1} \cup \post{\supp{\sigma}}$, and so,
  since $\supp{C_1} \cup \post{\supp{\sigma}}= \supp{D''}$, by ($a'$)
  and ($b'$), we get $\supp{D} \subseteq \supp{D''}$. Thus, since $t$
  is enabled at $D$ and, by the definition of $c$, it involves at most
  $c - 1$ agents, $t$ is also enabled in $(c - 1) \cdot
  D''$. Moreover, by ($d'$) we have $$c \cdot C_1'' \trans{t_1^{c^m}
  t_2^{c^{m-1}} \cdots t_m^c} c \cdot D''.$$ So, since $c > 1$, we
  obtain $$c \cdot C_1'' \trans{t_1^{c^m} t_2^{c^{m-1}} \cdots t_m^c}
  c \cdot D'' \trans{t} (D'' + E)$$ for some configuration $E$.
  Taking $C_1' \defeq c \cdot C_1''$, $C_2' \defeq D'' + E$, and $\pi'
  = t_1^{c^m} t_2^{c^{m-1}} \cdots t_m^c t$ we have $C_1' \trans{\pi'}
  C_2'$. We prove that $C_1'$, $C_2'$, and $\pi'$ satisfy ($a$),
  ($b$), and ($d$):\medskip

  \noindent($a$) We must show $\supp{C_1'} = \supp{C_1}$. It follows
  from
    $$
      \supp{C_1'} = \supp{c \cdot C_1''} = \supp{C_1''}
      \stackrel{(a')}{=} \supp{C_1}.
    $$

  \noindent($b$) We must show $\supp{C_2'} = \supp{C_1} \cup
  \post{\supp{\pi}}$. It follows from
  \begin{multline*}
    \supp{C_2'} = \supp{D'' + E} = \supp{D''} \cup \supp{E}
    \stackrel{(b')}{=} \supp{C_1''} \cup \post{\supp{\sigma}} \cup
    \supp{E} = \supp{C_1''} \cup \post{\supp{\sigma}} \cup \post{t} = \\
    \supp{C_1''} \cup \post{\supp{\sigma} \cup \{t\}} = \supp{C_1''}
    \cup \post{\supp{\sigma t}} \stackrel{(a')}{=} \supp{C_1} \cup
    \post{\supp{\pi}}.
  \end{multline*}

  \noindent($d$) We must show that $\pi' = t_1^{c^{m-1}}
  t_2^{c^{m-2}} \cdots t_m^c t$, where $t_1, \ldots, t_m, t$ belong to
  $\supp{\pi}$, and $m+1 \leq |\post{\{t_1, \ldots, t_m, t\}}|$. \\

  Since $t_1, \ldots, t_m$ belong to $\supp{\sigma}$ by~($d'$), the
  transitions $t_1, \ldots, t_m, t$ belong to $\supp{\sigma t} =
  \supp{\pi}$. Further, we have $m + 1 \leq |\post{\{t_1, \ldots,
    t_m\}}| + 1$ by~($d'$), and $|\post{\{t_1, \ldots, t_m\}}| + 1
  \leq |\post{\{t_1, \ldots, t_m, t\}}|$ because, by assumption,
  $\post{t}\not\subseteq \post{\{t_1, \ldots, t_m\}}$.
\end{proof}

\propRackoff*

\begin{proof}
  Let $Q = \{q_1, q_1, \ldots, q_n\}$, and let $b$ be a fresh symbol
  not contained in $Q$. We associate to
  $\PP$ a set $A \subseteq \Z^{|Q|+|T|+1}$. 
  The set $A$ contains two vectors $\vec{v}_t^1, \vec{v}_t^2$ 
  for every transition $t \in T$, defined as functions 
  $Q\cup T \cup \set{b} \rightarrow \Z$ in the following way:
  $\vec{v}_t^1(q) = - \prem{t}(q)$ for all $q \in Q$, 
  $\vec{v}_t^1(t') = 1 \mbox{ if $t=t'$ else } 0$ for all $t \in T$, 
  and $\vec{v}_t^1(b) = -1$; $\vec{v}_t^2(q) = \postm{t}(q)$ for all $q \in Q$, 
  $\vec{v}_t^2(t') = -1 \mbox{ if $t=t'$ else } 0$, and $\vec{v}_t^2(b) = 1$. 
  Intuitively, $\vec{v}_t^1$ ``removes'' agents from their current states, and
  $\vec{v}_t^2$ ``adds'' them to their new states. 
  It is easy to see that for every $\vec{v} \in \N^{|Q|+|T| + 1}$ 
  satisfying $\vec{v}(t)=0$ for every $t \in T$ and $\vec{v}(b) = 1$, 
  the VAS $(A, \vec{v})$ simulates $\PP$ from the
  configuration $C$ satisfying $C(q)=\vec{v}(q)$ for every $q \in
  Q$. An occurence of $t$ in $\PP$ is simulated by first adding
  $\vec{v}_t^1$ and then $\vec{v}_t^2$. The $b$-component 
  ensures that  $\vec{v}_t^1$ always directly precedes  $\vec{v}_t^2$. 
  Since $A$ contains $2|T|$
  vectors of dimension $\left(|Q|+|T| + 1\right)$ with entries taken from
  $\{-2, -1,0,1,2\}$, its size is bounded by $12|Q|^8$:
  \begin{align*}
  \size{A} & = \sum_{\vec{v} \in A} \size{\vec{v}} \\
           & = \sum_{\vec{v} \in A} \sum_{1 \leq i \leq (|Q|+|T| + 1)} \size{\max(|\vec{v}(i)|,1)} \\
           & \leq \sum_{\vec{v} \in A} 2 \cdot (|Q|+|T| + 1) \\
           & = 4|T| \cdot (|Q| + |T| + 1)\\
           & \leq 4|Q|^4 \cdot (|Q|^4 + |Q|^4 + |Q|^4)\\
           & = 12|Q|^8 & \text{}
  \end{align*}
  By applying Theorem~\ref{thm:rackoff} on $A$, we obtain the desired bound.
  %% Observe that the
  %% vector $\vec{v}$ to cover is the one satisfying $\vec{v}(q) = 1$ and
  %% $\vec{v}(q') = 0$ for every other state $q'$, and so it has size
  %% $|Q|$.
\end{proof}

\thmOneawareleader

\begin{proof}
  Let $\PP = (Q, T, \{q_0\}, L, O)$ be a 1-aware 2-way population
  protocol computing the predicate $x \geq n$. Let $q_0$ be the only
  initial state of $\PP$, and let $Q_1 \subseteq Q$ be the set of
  states of $\PP$ that make it 1-aware. By Proposition
  ~\ref{prop:Rackoff}, some state $q_1 \in Q_1$ is coverable from
  $\multiset{n \cdot q_0}+L$ by means of an execution $\sigma$ of
  length $2^{(3m)^m-1}$ where $m \leq 12 |Q|^8$.

  Let $k \defeq 2^{(3m)^m}$. Since $\sigma$ removes at most $k$ agents
  from state $q_0$, it is also enabled at the initial configuration
  $C_0' \defeq \multiset{k \cdot q_0} \mplus L$. Further, since
  $q_0 \notin Q_1$ and $q_1 \in Q_1$, we have $C_0' \trans{\sigma} C'$
  for some configuration $C'$ such that $C'(q_1) > 0$. By definition
  of 1-awareness, $O(C_0') = 1$, and thus since $\PP$ computes $x \geq
  n$, we have $k \geq n$.

  Therefore, $n \leq k \leq 2^{(3m)^m}$, which implies that $n \leq
  2^{(36 |Q|^8)^{12|Q|^8}} = 2^{2^{\log (36 |Q|^8) 12 |Q|^8}}$, and in
  turn that $\log\log (n) \leq \log (36 |Q|^8) \cdot 12 |Q|^8$ for
  every $n \geq 2$. Note that $\log(a) \leq \lambda \cdot a^{1
  / \lambda}$ for every $a, \lambda \in \N_{>0}$. Thus, by taking $a =
  36 |Q|^8$ and $\lambda = 8$, we obtain $\log\log n \leq 12 \cdot
  8 \cdot 36^{1/8} \cdot |Q|^{9} \leq 151 |Q|^{9}$, which implies that
  $|Q| \geq (\log \log(n) / 151)^{1 / 9}$.
\end{proof}

%% file: app-sys-lin.tex
%% Local macros
\newcommand{\istate}[3]{\bm{\mathrm{{#1}_{\mathnormal{#3}}^{#2}}}}

Since linear inequalities are subsumed by systems of linear
inequalities, we only give a proof sketch of
Theorem~\ref{thm:lin:ineq} and we instead focus on proving
Theorem~\ref{thm:sys:lin} in details.

\subsection{Linear inequalities}

\thmLinIneq*

\begin{proof}[Proof sketch]
  The bounds follow from the definition of $\PP_\text{lin}$ and
  Lemma~\ref{lem:kway}. Let us sketch the correctness of
  $\PP_\text{lin}$. We associate a value to each state in the natural
  way, \ie $\val(\state{x_i}) \defeq a_i$, $\val(\state{+2^i}) \defeq
  2^i$, $\val(\state{-2^i}) \defeq -2^i$ and $\val(\state{+0}) =
  \val(\state{-0}) \defeq 0$. Let $$X^+ \defeq \{x_i : 1 \leq i \leq
  k, a_i > 0\}\ \text{ and }\ X^- \defeq \{x_i : 1 \leq i \leq k, a_i <
  0\}.$$ For every configuration $C$, we let
  \begin{align*}
    \val^+(C) &\defeq \sum_{q \in Q^+ \cup X^+} \val(q) \cdot C(q), &
    \val^-(C) \defeq \sum_{q \in Q^- \cup X^-} \val(q) \cdot C(q),
  \end{align*}
  and $\val(C) \defeq \val^+(C) + \val^-(C)$.

  For every initial configuration $C_0$ and sequence $C \trans{\sigma}
  C'$, it can be shown that:
  \begin{itemize}
  \item $\val(C_0) = \sum_{1 \leq i \leq k} a_i \cdot C_0(\state{x_i})
    + c$,

  \item $\val^+(C) \geq \val^+(C')$, $\val^-(C) \leq \val^-(C')$ and
    $\val(C) = \val(C')$, and

  \item $C'(x) = C(x) - \sum_{r \in R} |\sigma|_{\mathrm{add}_{x, r}}$
    for every $x \in X$.
  \end{itemize}
  Using these facts, it is possible to show that the number of agents
  in the largest powers of $2$ cannot grow too much, as otherwise the
  represented value would be too large or too small:
  \begin{align*}
    C(\state{+2^n}) &\leq 1 + \sum_{x \in X^+} \sum_{r \in R}
    |\sigma|_{\mathrm{add}_{x,r}}\ \text{ and }\
    C(\state{-2^n}) \leq 1 + \sum_{x \in X^-} \sum_{r \in R}
    |\sigma|_{\mathrm{add}_{x,r}}.
  \end{align*}
  Combining these observations, and by using transitions of the form
  $\mathrm{up}_i^+$ and $\mathrm{up}_i^-$, it can be shown that
  \begin{itemize}
  \item If $C_0$ is initial and $C_0 \trans{*} C$, then there exist
    $C'$ and $r \in R$ s.t. $C \trans{*} C'$ and $C'(r) \geq n$.
  \end{itemize}
  
  This implies that, in any fair execution, transitions of the form
  $\mathrm{add}_{x,r}$ can occur until the number of agents in $X$
  stabilizes to $0$. Moreover, it implies that, in any fair execution,
  transitions of the form $\mathrm{down}_{i,r}^+$,
  $\mathrm{down}_{i,r}^-$ and $\mathrm{cancel}_i$ can occur until the
  number of agents in $Q^+$ or $Q^-$ stabilizes to $0$. Finally,
  ``signal'' transitions ensure that every fair execution stabilizes
  to the right output.
\end{proof}

\subsection{Conjunction of linear inequalities}

Let $A \in \Z^{m \times k}$ and $\vec{c} \in \Z^m$. Let us now
introduce in details the population protocol $\PP_{\text{sys}} = (Q,
T, I, L, O)$ for the predicate $A\vec{x} + \vec{c} >
\vec{0}$. Let $$b_\text{max} \defeq \max(1, \max\{|A_{i,j}| : i \in
    [m], j \in [k]\}, \max\{|c_i| : i \in [m]\})$$ and $n \defeq
    \lceil \log 2m^2 \rceil + \size{b_\text{max}}$. The following will
    later be crucial:
\begin{align}
  2^n > 2m^2 \cdot b_\text{max}.\label{eq:bmax}
\end{align}
The states of the protocol are defined as $Q \defeq X \cup Q^+ \cup
Q^- \cup R$ where
\begin{align*}
  X   &\defeq \{\state{x_1}, \state{x_2}, \ldots, \state{x_k}\}, &
  R   &\defeq \{\istate{0}{}{0}, \istate{0}{}{1}\}, \\
  Q_j^+ &\defeq \{\istate{+2}{i}{j,\alpha} : i \in [0, n], \alpha \in
  \{0, 1\}\}, &
  Q_j^- &\defeq \{\istate{-2}{i}{j}  : i \in [0, n]\}, \\
  Q^+ &\defeq Q_1^+ \cup Q_2^+ \cup \cdots \cup Q_m^+, &
  Q^- &\defeq Q_1^- \cup Q_2^- \cup \cdots \cup Q_m^-.
\end{align*}
The initial states are defined as $I \defeq X$, and the output mapping as
$$
O(q) \defeq 
\begin{cases}
  1 & \text{if } q = \istate{0}{}{1}, \text{ or } q = \istate{+2}{i}{j,1}
  \text{ for some } i \in [0, n], j \in [m] \\
  0 & \text{otherwise}.
\end{cases}
$$

In order to define leaders and transitions, let us first give some
definitions. Let $\mathrm{rep}_j(d) \colon \Z \to \pop{Q \setminus X}$
be defined as follows:
\begin{align*}
  \mathrm{rep}_j(d) &\defeq
  \begin{cases}
    \multiset{\istate{+2}{i}{j,0} : i \in \bits{d}}   & \text{if } d > 0, \\
    \multiset{\istate{-2}{i}{j\phantom{,0}} : i \in \bits{|d|}} & \text{if } d < 0, \\
    \multiset{\state{-0}}                             & \text{if } d = 0.
  \end{cases}
\end{align*}
Let $\rep{\vec{d}} \colon \Z^m \to \pop{Q \setminus X}$ be defined as
$\rep{\vec{d}} \defeq \mathrm{rep}_1(\vec{d}_1) \mplus
\mathrm{rep}_2(\vec{d}_2) \mplus \ldots \mplus
\mathrm{rep}_m(\vec{d}_m)$. Leaders are defined as $L \defeq
\rep{\vec{c}} \mplus \multiset{(5mn + 1) \cdot \istate{0}{}{0}}$.

It remains to describe the set of transitions $T$. It contains the
following transitions which allow to change representations of numbers
over $Q^+ \cup Q^-$:
\begin{align*}
  \mathrm{up}_{i,j,\alpha,\beta}^+:\ & \istate{+2}{i}{j,\alpha},
  \istate{+2}{i}{j,\beta} \mapsto \istate{+2}{i+1}{j,\alpha \land
    \beta}, \istate{0}{}{\alpha \land \beta} &
  \mathrm{down}_{i+1,j,\alpha,\beta}^+:\ & \istate{+2}{i+1}{j,\alpha},
  \istate{0}{}{\beta} \mapsto \istate{+2}{i}{j,\alpha \land \beta},
  \istate{+2}{i}{j,\alpha \land \beta} \\
  \mathrm{up}_{i,j}^-:\ & \istate{-2}{i}{j}, \istate{-2}{i}{j} \mapsto
  \istate{-2}{i+1}{j}, \istate{0}{}{0} &
  \mathrm{down}_{i+1,j,\alpha}^-:\ & \istate{-2}{i+1}{j},
  \istate{0}{}{\alpha} \mapsto \istate{-2}{i}{j}, \istate{-2}{i}{j}
\end{align*}
where $i \in [0, n-1]$, $j \in [m]$ and $\alpha, \beta \in \{0,
1\}$. It contains the following transitions to cancel out equal
numbers:
\begin{align*}
  \mathrm{cancel}_{i,j,\alpha}:\ & \istate{+2}{i}{j,\alpha},
  \istate{-2}{i}{j} \mapsto \istate{0}{}{\alpha}, \istate{0}{}{0}
\end{align*}
where $i \in [0, n]$, $j \in [m]$ and $\alpha \in \{0, 1\}$. It
contains the following transitions to signal false and true consensus:
\begin{align*}
  \mathrm{false}_{i,j}^+:\ & \istate{0}{}{0}, \istate{+2}{i}{j,1}
  \mapsto \istate{0}{}{0}, \istate{+2}{i}{j,0} &
  \mathrm{false}:\ & \istate{0}{}{0}, \istate{0}{}{1} \mapsto
  \istate{0}{}{0}, \istate{0}{}{0} \\
  \mathrm{false}_{i,j}^-:\ & \istate{-2}{i}{j}, \istate{0}{}{1}
  \mapsto \istate{-2}{i}{j}, \istate{0}{}{0}
\end{align*}
\begin{align*}
  \mathrm{true}:\ & \istate{+2}{0}{1,0}, \istate{+2}{0}{2,0}, \ldots,
  \istate{+2}{0}{m,0}, \istate{0}{}{0} \mapsto \istate{+2}{0}{1,1},
  \istate{+2}{0}{2,1}, \ldots, \istate{+2}{0}{m,1}, \istate{0}{}{1} \\
  \mathrm{true}_{i,j}:\ & \istate{0}{}{1}, \istate{+2}{i}{j,0} \mapsto
  \istate{0}{}{1}, \istate{+2}{i}{j,1}
\end{align*}
where $i \in [0, n]$ and $j \in [m]$. Finally, it contains the
following transitions to convert variables to their coefficients:
\begin{align*}
  \mathrm{add}_{j, \alpha}:\ &\state{x_j},
  \underbrace{\istate{0}{}{\alpha}, \istate{0}{}{\alpha}, \ldots,
    \istate{0}{}{\alpha}}_{|\rep{A_{\star,j}}|- 1 \text{
        times}} \mapsto \rep{A_{\star,j}}
\end{align*}
where $j \in [k]$, $\alpha \in \{0, 1\}$, and $A_{\star,j}$ is the
$j^\text{th}$ column of $A$.

The rest of this appendix is dedicated to proving the correctness of
$\PP_\text{sys}$. Before doing so, we need to introduce additional
definitions. Let $\val \colon Q \to \N$ be the function that
associates a value to each state as follows:
\begin{align*}
  \val(\state{x_j}) &\defeq \sum_{i \in [m]} A_{i,j} && \text{for every } j
  \in [k], \\
  \val(\istate{+2}{i}{j,\alpha}) &\defeq 2^i && \text{for every } i \in [0,
    n], j \in [m], \alpha \in \{0, 1\}, \\
  \val(\istate{-2}{i}{j}) &\defeq -2^i && \text{for every } i \in [0, n], j
  \in [m] \\
  \val(\istate{0}{}{0}) &\defeq \val(\istate{0}{}{1}) \defeq 0.
\end{align*}
We extend $\val$ to configurations. For every $C \in \pop{Q}$ and
every $i \in [m]$, let
\begin{align*}
  \val_i^+(C) &\defeq \sum_{q \in Q_i^+} \val(q) \cdot C(q), &
  \val^+(C) &\defeq \sum_{i \in [m]} \val_i^+(C), \\
  \val_i^-(C) &\defeq \sum_{q \in Q_i^-} \val(q) \cdot C(q), &
  \val^-(C) &\defeq \sum_{i \in [m]} \val_i^-(C), \\
  \val_i(C) &\defeq \val_i^+(C) + \val_i^-(C) + \sum_{j \in [k]}
  A_{i,j} \cdot C(\state{x_j}) &
  \val(C)   &\defeq \sum_{i \in [m]} \val_i(C).
\end{align*}
For every $j \in [k]$ and $\sigma \in T^*$, let $b_j \defeq \sum_{i
  \in [m]} |A_{i, j}|$ and let $\mathrm{num}_i(\sigma) \defeq
|\sigma|_{\mathrm{add}_{i,0}} + |\sigma|_{\mathrm{add}_{i,1}}$. It is
not so difficult to derive the following properties from the above
definitions:
\begin{proposition}\label{prop:threshold:val2}
  Let $C_0, C, C' \in \pop{Q}$ and $\sigma \in T^*$ be such that $C_0$
  is initial and $C \trans{\sigma} C'$. The following holds for every
  $i \in [m]$:
  \begin{enumerate}
  \item[(a)] $\val_i^+(C_0) = \max(\vec{c}_i, 0)$ and $\val_i^-(C_0) =
    \min(\vec{c}_i, 0)$,

  \item[(b)] $\val_i(C') = \val_i(C)$,

  \item[(c)] if $C(X) = 0$, then $\val_i^+(C) \geq \val_i^+(C')$,

  \item[(d)] if $C(X) = 0$, then $\val_i^-(C) \leq \val_i^+(C')$,

  \item[(e)] $\val^+(C') \leq \val^+(C) + \sum_{j \in [k]}
    \mathrm{num}_j(\sigma) \cdot b_j$,

  \item[(f)] $\val^-(C') \geq \val^-(C) - \sum_{j \in [k]}
    \mathrm{num}_j(\sigma) \cdot b_j$,

  \item[(g)] $C'(X) = C(X) - \sum_{j \in [k]} \mathrm{num}_j(\sigma)$.
  \end{enumerate}
\end{proposition}

From Proposition~\ref{prop:threshold:val2}, we obtain the following
useful proposition:

\begin{proposition}\label{prop:threshold:topbit2}
  Let $C_0, C \in \pop{Q}$ and $\sigma \in T^*$ be such that $C_0$ is
  initial and $C_0 \trans{\sigma} C$. For every $i \in [m]$, the
  following holds: $$C(\{\istate{+2}{n}{i,0}, \istate{+2}{n}{i,1}\})
  \leq (d + 1) / 2m \text{ and } C(\istate{-2}{n}{i}) \leq (d + 1) /
  2m$$ where $d = \sum_{j \in [k]} \mathrm{num}_j(\sigma)$.
\end{proposition}

\begin{proof}
  Let $i \in [m]$. We only prove the first claim, the second one
  follows symmetrically. Let $S \defeq \{\istate{+2}{n}{i,0},
  \istate{+2}{n}{i,1}\}$. We must show that $C(S) \leq (d + 1) /
  2m$. For the sake of contradiction, suppose $C(S) > (d + 1) /
  2m$. We derive the following contradiction:
  \begin{align*}
    \val^+(C)
    &\geq 2^n \cdot C(S) && \text{(by def. of $\val^+$)} \\
    &> 2^n \cdot ((d + 1) / 2m) && \text{(by assumption)} \\
    &= \frac{2^n + \sum_{j \in [k]} (2^n \cdot
      \mathrm{num}_j(\sigma))}{2m} && \text{(by def. of $d$)} \\
    &> \frac{2m^2 \cdot b_\text{max} + \sum_{j \in [k]}
      2m^2 \cdot b_\text{max} \cdot \mathrm{num}_j(\sigma)}{2m} &&
    \text{(by~\eqref{eq:bmax})} \\
    &= m \cdot b_\text{max} + \sum_{j \in [k]} m \cdot b_\text{max}
    \cdot \mathrm{num}_j(\sigma) && \\
    &\geq \val^+(C_0) + \sum_{j \in [k]} m \cdot b_\text{max} \cdot
    \mathrm{num}_j(\sigma) && \text{(by
      Prop.~\ref{prop:threshold:val2}(a))} \\
    &\geq \val^+(C_0) + \sum_{j \in [k]} b_j \cdot
    \mathrm{num}_j(\sigma) && \text{(by def. of $b_\text{max}$ and $b_j$)}\\
    &\geq \val^+(C) && \text{(by Prop.~\ref{prop:threshold:val2}(e))}.\qedhere
  \end{align*}
\end{proof}

The following proposition shows that is always possible to convert at
least $mn$ agents back to a state of $R$. This will later be useful in
arguing that the number of agents in $X$ can eventually be decreased
to zero.

\begin{proposition}\label{prop:threshold:normalize2}
  Let $C_0, C \in \pop{Q}$ be such that $C_0$ is initial. If $C_0
  \trans{*} C$, then there exist a configuration $C'$ and $\alpha \in
  \{0, 1\}$ such that $C \trans{*} C'$ and $C'(\istate{0}{}{\alpha})
  \geq mn$.
\end{proposition}

\begin{proof}
  If $C(R) \geq 2mn$, then $C' \defeq C$ satisfies the claim by the
  pigeonhole principle. Therefore, assume $C(R) < 2mn$. Let $\sigma
  \in T^*$ be such that $C_0 \trans{\sigma} C$. Let $$U \defeq
  \{\istate{+2}{n}{j,0}, \istate{+2}{n}{j,1}, \istate{-2}{n}{j} : j
  \in [m]\}\ \text{ and }\ V \defeq (Q^+ \cup Q^-) \setminus U.$$ We have
  \begin{align*}
    C(V) &= |C_0| - C(U) - C(X) - C(R) && \text{(by
      $|C| = |C_0|$)} \\
    &> |C_0| - C(U) - C(X) - 2mn && \text{(by
      assumption)} \\
    &= |C_0| - C(U) - (C_0(X) - \sum_{j \in [k]} \mathrm{num}_j(\sigma)) -
    2mn && \text{(by Prop.~\ref{prop:threshold:val2}(g))} \\
    &\geq 3mn + 1 - C(U) + \sum_{j \in [k]} \mathrm{num}_j(\sigma) &&
    \text{(by $C_0(R) \geq 5mn + 1$)} \\
    &\geq 3mn + 1 - 2m\left(\frac{1 + \sum_{j \in [k]}
      \mathrm{num}_j(\sigma)}{2m}\right) + \sum_{j \in [k]}
    \mathrm{num}_j(\sigma) && \text{(by
      Prop.~\ref{prop:threshold:topbit2})} \\
    &= 3mn.
  \end{align*}
  Since $C(V) > 3mn = |V|$, the pigeonhole principle implies that
  $C(q) \geq 2$ for some $q \in V$. Therefore, a transition of the
  form $\mathrm{up}_{i,j,\alpha,\beta}^+$ or $\mathrm{up}_{i,j}^-$ can
  occur from $C$, leading to a configuration $D$ such that $D(R) =
  C(R) + 1$. If $D(R) < 2mn$, then this argument can be repeated until
  a configuration $C'$ such that $C'(R) \geq 2mn$ is reached.
\end{proof}

We now show that, in any fair execution, the number of agents in $X$
eventually stabilizes to 0, and the value associated to each conjunct
stabilizes to either some positive or some negative number.

\begin{proposition}\label{prop:threshold:stabilize2}
  Let $\pi = C_0 C_1 \cdots$ be a fair execution from an initial
  configuration $C_0$. There exist $\ell \in \N$, $d_1^+, d_2^+,
  \ldots, d_m^+ \geq 0$ and $d_1^-, d_2^-, \ldots, d_m^- \leq 0$ such
  that for every $i \in [m]$, the following holds:
  \begin{enumerate}
    \item $C_\ell(X) = C_{\ell+1}(X) = \cdots = 0$,
    \item $\val_i^+(C_\ell) = \val_i^+(C_{\ell+1}) = \cdots = d_i^+$,
    \item $\val_i^-(C_\ell) = \val_i^-(C_{\ell+1}) = \cdots = d_i^-$,
    \item $d_i^+ = 0$ or $d_i^- = 0$.
  \end{enumerate}
\end{proposition}

\begin{proof}
  For the sake of contradiction, assume there exist infinitely many
  indices $i$ such that $C_i(X) > 0$. Let $i \in \N$ be one of these
  indices. By Proposition~\ref{prop:threshold:normalize2}, there exist
  $D_i \in \pop{Q}$ and $\alpha \in \{0, 1\}$ such that $C_i \trans{*}
  D_i$ and $D_i(\istate{0}{}{\alpha}) \geq mn$. Hence, by definition
  of $T$, there exists $j \in [k]$ such that $\mathrm{add}_{j,\alpha}$
  is enabled at $D_j$. Since this holds for infinitely many indices,
  fairness implies that one transition of $\{\mathrm{add}_{j,\alpha} :
  j \in [k], \alpha \in \{0, 1\}\}$ is taken infinitely often along
  $\pi$. This is impossible since the number of agents in $X$ cannot
  increase, and thus would eventually drop below zero. Therefore,
  there exists $\ell \in \N$ such that $C_\ell(X) = C_{\ell+1}(X) =
  \cdots = 0$.

  Let $i \in [m]$. By Proposition~\ref{prop:threshold:val2}(c,d,g), we
  have
  \begin{align*}
    \val_i^+(C_\ell) &\geq \val_i^+(C_{\ell+1}) \geq \cdots \geq 0, \\
    \val_i^-(C_\ell) &\leq \val_i^-(C_{\ell+1}) \leq \cdots \leq 0.
  \end{align*}
  Therefore, there exist $\ell' \geq \ell$, $d_i^+ \geq 0$ and $d_i^-
  \leq 0$ such that
  \begin{align}
    \val_i^+(C_{\ell'}) = \val_i^+(C_{\ell'+1}) = \cdots =
    d_i^+,\label{eq:stabp2} \\
    \val_i^-(C_{\ell'}) = \val_i^-(C_{\ell'+1}) = \cdots =
    d_i^-.\label{eq:stabm2}
  \end{align}
  For the sake of contradiction, assume that $d_i^+ \neq 0$ and $d_i^-
  \neq 0$. Let $J \subseteq [\ell', +\infty)$ be the set of all
    indices $j$ such that $D_j(Q_i^+) > 0$ and $D_j(Q_i^-) > 0$. We
    may assume that $J$ is infinite, as otherwise fairness would
    contradict~(\ref{eq:stabp2}) or~(\ref{eq:stabm2}). Let $j \in
    J$. There exist $\lambda, \lambda' \in [0, n]$ and $\alpha \in
    \{0, 1\}$ such that $D_j(\istate{+2}{\lambda}{i,\alpha}) > 0$ and
    $D_j(\istate{+2}{\lambda'}{i}) > 0$. Assume without loss of
    generality that $\lambda \geq \lambda'$. The other case is proven
    symmetrically. Since $D_j(\istate{0}{}{\beta}) \geq n \geq \lambda
    - \lambda'$ for some $\beta \in \{0, 1\}$, the
    sequence $$\mathrm{down}_{\lambda,i,\alpha,\beta}^+ \cdot
    \mathrm{down}_{\lambda-1,i,\alpha \land \beta,\beta}^+ \cdots
    \mathrm{down}_{\lambda'+1,i,\alpha \land \beta,\beta}^+ \cdot
    \mathrm{cancel}_{\lambda',i,\alpha \land \beta}$$ can occur from
    $D_j$. The resulting configuration $E_j$ is such that
    \begin{align*}
      \val_i^+(E_j) &< \val_i^+(D_i) \leq \val_i^+(C_i) = d_i^+, \\
      \val_i^-(E_j) &> \val_i^-(D_i) \geq \val_i^+(C_i) = d_i^-.
    \end{align*}
    Since $\{E_{\ell'}, E_{\ell'+1}, \ldots\}$ is finite, fairness
    implies that one of these configurations occurs infinitely often
    along $\pi$. This contradicts~(\ref{eq:stabp2})
    and~(\ref{eq:stabm2}).
\end{proof}

We are now ready to prove correctness of $\PP_{\text{sys}}$.

\begin{theorem}
  $\PP_{\text{sys}}$ is well-specified and correct.
\end{theorem}

\begin{proof}
  Let $\pi: C_0 \trans{\sigma_1} C_1 \trans{\sigma_2} \cdots$ be a
  fair execution from an initial configuration $C_0$. By
  Proposition~\ref{prop:threshold:stabilize2}, there exist $\ell \in
  \N$, $d_1^+, d_2^+, \ldots, d_m^+ \geq 0$ and $d_1^-, d_2^-, \ldots,
  d_m^- \leq 0$ such that for every $i \in [m]$, the following holds:
  \begin{itemize}
    \item[(a)] $C_\ell(X) = C_{\ell+1}(X) = \cdots = 0$,
    \item[(b)] $\val_i^+(C_\ell) = \val_i^+(C_{\ell+1}) = \cdots = d_i^+$,
    \item[(c)] $\val_i^-(C_\ell) = \val_i^-(C_{\ell+1}) = \cdots = d_i^-$,
    \item[(d)] $d_i^+ = 0$ or $d_i^- = 0$.
  \end{itemize}
  We first show well-specification. There are two cases to
  consider.\medskip

  \noindent\emph{Case 1: $d_i^+ > 0$ for every $i \in [m]$.} We claim
  that for every $j \geq \ell$, configuration $C_j$ can reach a
  configuration that contains some agent in state
  $\istate{0}{}{1}$. Let us argue that the validity of the claim
  concludes the case. By fairness, the claim implies that
  $C_j(\istate{0}{}{1}) > 0$ for infinitely many indices
  $j$. Therefore, by fairness and transitions of the form
  $\mathrm{true}_{\star,\star}$, we have $O(C_j) = 1$ for infinitely
  many indices $j$. By examining the presets and postsets of
  transitions from $T$, we observe that any configuration whose output
  is $1$ must be stable.

  Let us now prove the claim. Let $j \geq \ell$.  By
  Proposition~\ref{prop:threshold:normalize2}, there exist $D_j \in
  \pop{Q}$ and $\beta \in \{0, 1\}$ such that $C_j \trans{*} D_j$ and
  $D_j(\istate{0}{}{\beta}) \geq mn$. If $\beta = 1$, we are
  done. Thus, assume $\beta = 0$. Since $d_1^+, d_2^+, \ldots, d_m^+ >
  0$, there exist $\lambda_1, \lambda_2, \ldots, \lambda_m \in [0, n]$
  and $\alpha_1, \alpha_2, \ldots, \alpha_m \in \{0, 1\}$ such that
  $D_j(\istate{+2}{\lambda_i}{i, \alpha_i}) > 0$ for every $i \in
  [m]$. Since $D_j(\istate{0}{}{\beta}) \geq mn$, it is possible to
  construct a configuration $E_j \in \pop{Q}$ and sequence $w \in
  T^*$, made of transition ``$\mathrm{false}$'' and transitions of the
  form $\mathrm{down}_{\star,\star,\star,\star}^+$, such that
  \begin{itemize}
    \item $D_j \trans{w} E_j$, and
    \item $E_j(\istate{+2}{0}{i,0}) > 0$ for every $i \in [m]$.
  \end{itemize}
  Thus, transition ``$\mathrm{true}$'' can occur from $E_j$, leading
  to a configuration $F_j$ such that $F_j(\istate{0}{}{1}) >
  0$.\medskip

  \noindent\emph{Case 2: $d_i^+ = 0$ for some $i \in [m]$.} We claim
  that for every $j \geq \ell$, configuration $C_j$ can reach a
  configuration that contains some agent in state
  $\istate{0}{}{0}$. Let us argue that the validity of the claim
  concludes the case. By fairness, the claim implies that
  $C_j(\istate{0}{}{0}) > 0$ for infinitely many indices
  $j$. Therefore, by fairness, transition ``$\mathrm{false}$'' and
  transitions of the form $\mathrm{false}_{\star,\star}^\star$, we
  have $O(C_j) = 0$ for infinitely many indices $j$. By examining the
  presets and postsets of transitions from $T$, we observe that a
  configuration whose output is $0$ can only reach a configuration
  whose output is not $0$ through transition
  ``$\mathrm{true}$''. Since $d_i^+ = 0$, we have $C_j(Q_i^+) = 0$ for
  every $j \geq \ell$. Therefore, transition ``$\mathrm{true}$'' is
  disabled at $C_j$ for every $j \geq \ell$.

  Let us now prove the claim. Let $j \geq \ell$. By
  Proposition~\ref{prop:threshold:normalize2}, there exist $D_j \in
  \pop{Q}$ and $\beta \in \{0, 1\}$ such that $C_j \trans{*} D_j$ and
  $D_j(\istate{0}{}{\beta}) \geq mn$. If $\beta = 0$, we are
  done. Thus, assume $\beta = 1$. If $C_j(Q^-) > 0$, then we are done,
  since a transition of the form $\mathrm{false}_{\star, \star}^-$ can
  occur, leading to a configuration $E_j$ such that
  $E_j(\istate{0}{}{0}) > 0$. Therefore, assume $C_j(Q^-) = 0$. Since
  $C_j(\istate{0}{}{1}) > 0$, the prefix $\sigma_1 \sigma_2 \cdots
  \sigma_j$ must contain the transition ``$\mathrm{true}$''. Thus,
  there exists $j' < j$ such that $C_{j'}(Q_i^+) > 0$. Let $j'$ be the
  largest such index. Transition $\sigma_{j'+1}$ must be of the form
  $\mathrm{cancel}_{\star,\star,\star}$. Therefore,
  $C_{j'+1}(\istate{0}{}{0}) > 1$. By inspection of $T$, we observe
  that ``$\mathrm{true}$'' is the only transition that can decrease
  the number of agents in $\istate{0}{}{0}$. By maximality of $j'$, we
  have $C_{j'+1}(Q_i^+) = C_{j'+2}(Q_i^+) = \cdots = 0$. Thus
  transition ``$\mathrm{true}$'' cannot occur, and hence
  $C_j(\istate{0}{}{0}) > 0$.\medskip

  We are done proving well-specification. To conclude the proof, let
  us argue that $\PP_{\text{sys}}$ indeed computes the predicate
  $A\vec{x} + \vec{c} > \vec{0}$. Let $j \geq \ell$ be such that
  $C_j$ is stable. For every $i \in [m]$, we have
  \begin{align*}
    \vec{c}_i + \sum_{j \in [k]} A_{i,j} \cdot C_0(\state{x_j})
    &= \val_i(C_0) && \text{(By Prop.~\ref{prop:threshold:val2}(a))} \\
    &= \val_i(C_j) && \text{(By Prop.~\ref{prop:threshold:val2}(b))} \\
    &= \val_i^+(C_j) + \val_i^-(C_j) + 0\\
    &= d_i^+ + d_i^-.
  \end{align*}
  Recall that $d_i^+ \geq 0$ and $d_i^- \leq 0$ for every $i \in
  [m]$. If $A\vec{x} + \vec{c} > \vec{0}$ holds, then we must have
  $d_i^+ > 0$ for every $i \in [m]$. Therefore, case 1 holds, and
  hence $O(\pi) = O(C_j) = 1$, which is correct. If $A\vec{x} +
  \vec{c} > \vec{0}$ does not hold, then we must have $d_i^+ = 0$ for
  some $i \in [m]$. Therefore, case 2 holds, and hence $O(\pi) =
  O(C_j) = 0$, which is also correct.
\end{proof}

We may now prove the theorem from the main text:

\thmSysLin*

\begin{proof}
  The value $n$ which occurs in the statement of the theorem differs
  from the $n$ defined in this appendix. To avoid any confusion, let
  us rename the latter as $\ell$, \ie $\ell \defeq \lceil \log 2m^2
  \rceil + \size{b_\text{max}}$. Protocol $\PP_\text{sys}$ has $|Q| =
  3m(\ell+1) + k + 2$ states. Among these states, one transition is
  $(m+1)$-way and $k$ transitions are $\ell$-way. By applying
  Lemma~\ref{lem:kway}, we obtain a 2-way population protocol
  $\PP_\text{sys}'$ which computes the same predicate as
  $\PP_\text{sys}$ and whose number of states $|Q'|$ is bounded as
  follows:
  \begin{align*}
    |Q'|
    &= |Q| + 3(m+1) + 3k\ell && \text{(By Lemma~\ref{lem:kway})} \\
    &= [3m(\ell+1) + k + 2] + 3(m+1) + 3k\ell && \text{(by the size of $Q$)} \\
    &= [3m\ell + 3m + k + 2] + [3m + 3] + 3k\ell \\
    &= 3m\ell + 6m + 3k\ell + k + 5 \\
    &\leq 3m\ell + 6m\ell + 3k\ell + k\ell + 5k\ell \\
    &= 9m\ell + 9k\ell \\
    &= 9\ell(m + k) \\
    &= 9[\lceil \log 2m^2 \rceil + \size{b_\text{max}}](m + k) &&
    \text{(by def. of $\ell$)} \\
    &= 9[\lceil \log 2 + 2 \log m \rceil + \size{b_\text{max}}](m + k) \\
    &= 9[\lceil 1 + 2 \log m \rceil + \size{b_\text{max}}](m + k) \\
    &\leq 9[2 + 2\log m + \size{b_\text{max}}](m + k) \\
    &= 9(2 + 2\log m + n)(m + k) && \text{(by $n = \size{b_\text{max}}$)} \\
    &\leq 9(3\log m + 3n)(m + k) \\
    &= 27(\log m + n)(m + k).
  \end{align*}
  Moreover, the number of leaders of $\PP_\text{sys}'$ is the same as
  for $\PP_\text{sys}$, namely
  \begin{align*}
    |L|
    &= 5m\ell + 1 + |\rep{\vec{c}}| && \text{(by def. of $L$)} \\
    &= 5m[\lceil \log 2m^2 \rceil + \size{b_\text{max}}] + 1 +
    |\rep{\vec{c}}| && \text{(by def. of $\ell$)} \\
    &\leq 5m[2 \log m + 2 + \size{b_\text{max}}] + 1 + |\rep{\vec{c}}|
    && \text{(by $\lceil\log{2m^2}\rceil \leq 2\log{m} + 2$)} \\
    &\leq 5m[2 \log m + 2 + n] + 1 + mn && \text{(by $n =
      \size{b_\text{max}}$ and $|\rep{\vec{c}}| \leq m \cdot
      \size{b_\text{max}}$)} \\
    &= [10m\log m + 10m + 5mn] + 1 + mn \\
    &= 10m\log m + 10m + 6mn + 1 \\
    &= 10m\log m + 4m + 6m + 6mn + 1 \\
    &\leq 10m\log m + 4m\log m + 6mn + 6mn + mn \\
    &= 14m\log m + 13mn \\
    &\leq 14m(\log m + n). &&\qedhere
  \end{align*}
\end{proof}